\documentclass[]{elsarticle}
\usepackage{lineno}
\usepackage[latin1]{inputenc}
\usepackage{amsmath, amsthm, amsfonts, amssymb}
\usepackage{graphicx}

\usepackage{color}
\usepackage{tikz}
\usetikzlibrary{matrix}

\modulolinenumbers[5]
\usepackage{numcompress}
\biboptions{authoryear}

\newtheorem{theorem}{Theorem}
\newtheorem*{theorem*}{Theorem}
\newtheorem{corollary}{Corollary}
\newtheorem{lemma}{Lemma}
\newtheorem{proposition}{Proposition}

\newcommand\E{\mathbb E}
\newcommand\Pb{\mathbb P}
\newcommand\R{\mathbb R}
\newcommand\Z{\mathbb Z}

\newcommand\Mc{\mathcal M}
\newcommand\Hc{\mathcal H}

\newcommand\Dc{\mathcal D}

\newcommand\Xc{\mathcal X}

\newcommand\Nc{\mathcal N}
\newcommand\Fc{\mathcal F}

\begin{document}
\begin{frontmatter}
\title{The non-equilibrium allele frequency spectrum\\
 in a Poisson random field framework}
\date{\today}

\author[ik]{Ingemar Kaj\corref{cor1}}
\ead{ikaj@math.uu.se}

\author[cfm]{Carina F. Mugal}
\ead{carina.mugal@ebc.uu.se}

\cortext[cor1]{Corresponding author}

\address[ik]{Department of Mathematics, Uppsala University,
  Uppsala, Sweden}
\address[cfm]{Department of Ecology and Genetics, Uppsala University,
 Uppsala, Sweden}

\begin{abstract}
In population genetic studies, the allele frequency spectrum (AFS)
efficiently summarizes genome-wide polymorphism data and shapes a
variety of allele frequency-based summary statistics. While existing
theory typically features equilibrium conditions, emerging methodology
requires an analytical understanding of the build-up of the allele
frequencies over time. In this work, we use the framework of Poisson
random fields to derive new representations of the non-equilibrium AFS
for the case of a Wright-Fisher population model with selection. In
our approach, the AFS is a scaling-limit of the expectation of a
Poisson stochastic integral and the representation of the
non-equilibrium AFS arises in terms of a fixation time probability
distribution.  The known duality between the Wright-Fisher diffusion
process and a birth and death process generalizing Kingman's
coalescent yields an additional representation. The results carry over
to the setting of a random sample drawn from the population and
provide the non-equilibrium behavior of sample statistics. Our
findings are consistent with and extend a previous approach where the
non-equilibrium AFS solves a partial differential forward equation
with a non-traditional boundary condition. Moreover, we provide a
bridge to previous coalescent-based work, and hence tie several
frameworks together. Since frequency-based summary statistics are
widely used in population genetics, for example, to identify candidate
loci of adaptive evolution, to infer the demographic history of a
population, or to improve our understanding of the underlying
mechanics of speciation events, the presented results are potentially
useful for a broad range of topics.
\end{abstract}

\end{frontmatter}

\section{Introduction}

The allele frequency spectrum (AFS) describes the distribution of
allele frequencies over a large number of identical and independent
loci. In practice, the AFS is approximated by allele frequencies
recorded in a sample of individuals. Here, the recent progress in whole-genome
re-sequencing has significantly improved the accessibility of the AFS, and several allele
frequency-based summary statistics have become central measurements in population
genetic studies. The estimation of the AFS is then often based on polymorphic
nucleotide sites, where the frequency of the derived allele over a finite collection
of sites in the sample is summarized. In this context, the otherwise equivalent
term `site frequency spectrum' (SFS) is frequently used. For the purpose
of generality, we here use the term AFS.

The theory on the AFS was initiated in the 1930s with the classical
work of Fisher and Wright in a framework of diffusion theory including
effects of natural selection \citep{Fisher1930,Wright1931,Wright1938}.
Subsequently, \cite{Kimura1964} pioneered the systematic use of
stochastic processes in population genetics, and developed the theory
further. In particular, he considered the equilibrium distribution of
allele frequencies under irreversible mutation in an ensemble of polymorphic loci
\citep{Kimura1970a}. Central to these successful applications of diffusion
theory in describing the equilibrium limit AFS for various mutation
and selection scenarios is the Green function representation of
diffusion process occupation time functionals \citep{KarlinTaylor1981}.
Then, in order to study the impact of natural selection on the number
of fixations in diverging species, \cite{Sawyer1992} introduced the Poisson
random field framework. The basic assumptions of this approach are that
new mutant alleles arise at Poisson times, mutations are
irreversible, and the frequencies of the descendants of each
mutation are described by independent Markov processes (no
linkage). The loss or fixation of a mutant allele is captured by the
separate events of extinction or fixation of the Markov process. The
collection of Markov processes form a Poisson random field in the
sense that the limiting distributions of the allele frequencies are
independent Poisson random variables. In particular, the number of
fixations is a Poisson random variable with expected value increasing
linearly over time. Segregating mutations are, on the other hand, in
equilibrium with respect to time, and hence the marginal distributions
of the corresponding Poisson variables are stationary. In other words,
the AFS is assumed to be in equilibrium with respect to time.

More recently, \cite{Evans2007} initiated the study of the
non-equilibrium AFS in a single population including effects of
natural selection, in the sense of deriving a function $f(t,x)$ which
represents the expected fraction of alleles of frequency $x$ existing
at some time $t$, given an initial fraction $f(0,x)$ of alleles at
time $t=0$. Some of the modeling parameters, such as population size
and selection intensity, are also allowed to depend on time. The
resulting non-equilibrium AFS $f(t,x)$ is provided as a solution to a
partial differential equation (PDE), essentially the Kolmogorov
forward equation for the corresponding diffusion, linked to a given
rate of mutational influx via a specific boundary condition of
$f(t,x)$ as $x\to 0$. An additional approximation method using moments
is employed to study the resulting allele frequencies in a sample.
Building on this approach, \cite{Zivkovic2011} provide analytical
results on the non-equilibrium AFS for the neutral case, focusing on
time-dependence arising due to changes in population size.  In the
same direction, \cite{Zivkovic2015} consider the case of natural
selection and develop the moment approximation method for a scenario
of piecewise-constant population size starting from an equilibrium.

In a parallel methodological track the AFS has been studied using the
view of coalescent theory, where mutations are randomly placed on the
branches of a genealogy of a sample of individuals \citep{Kingman1982}.
First, \cite{Fu1995} obtained a representation of
the stationary AFS for a single population under the assumptions of
neutrality and constant population size, by deriving mean and variance
of the number of mutations on each branch of a given
length. \cite{Griffiths1998} explored the duality relation between the
neutral Wright-Fisher diffusion process and Kingman's pure death
coalescent process further and addressed deterministic changes in
population size. Moreover, \cite{Wakeley1997} obtained a description
of the joint AFS of two isolated populations descending from a common
ancestor under neutrality. \cite{Chen2012} elaborated on their work
and extended it to multiple populations and also modeled scenarios
such as selective sweeps, influx of migration and changes in population
size.

Here, we build on the work of \cite{Sawyer1992} and develop the
approach of \cite{Mugal2014} further to derive a
representation of the non-equilibrium AFS as the limiting expected
value of a suitable Poisson stochastic integral. The model is
developed in steps starting with finite population size $N$ and sequences
of $L$ sites, subject to mutational influx of derived alleles
and Wright-Fisher reproduction in discrete generation time. 
Assuming mutation and selection rates per individual and
generation of order $1/N$ and evolutionary time $t$ counting $Nt$
generations, we then apply the continuous time Wright-Fisher diffusion
approximation, but follow \cite{Evans2007} in keeping $N$ as a
modeling parameter. In a next stage of approximation 
the mutation rate per site tends to zero with preserved over-all
mutation rate for sequences, a procedure which we interpret and
implement as a limit in distribution as $L\to\infty$. The result is a
Poisson random field parametrized by $N$, which we study in some
detail. Then, we find the limiting expected values as $N\to\infty$ and
identify the time-dependent AFS which arises in the limit.
Thereby, we provide a link between the Poisson random field
approach by \cite{Sawyer1992} and the setting of \cite{Evans2007}, in
particular by identifying the PDE solution $f(t,x)$ in terms of a
Wright-Fisher fixation time probability distribution. An additional
representation is obtained by elaborating on the duality relation
between the Wright-Fisher diffusion process and a class of birth
and death processes, where birth rates are proportional to the
strength of selection \citep{Shiga1986,Athreya2005}.

\section{Poisson random field model}

\subsection{Basic Markov chain model} 

A population consists of $N$ individuals, where each individual is
represented by a sequence of $L$ sites. Random mutation events act on
sites, independently and uniformly over individuals, replacing an
ancestral allele by a derived. Only mono-allelic sites are affected by
mutation. Thus, the setting of the model only allows for two alleles,
the derived and the ancestral, in each site. The composition of
ancestral and derived alleles per site changes in discrete steps from
one generation to the next according to the Wright-Fisher reproduction
with selection, which relies on the following assumptions 1)
non-overlapping generations, 2) constant population size and 3) random
mating. The population dynamics is then given by a collection of
independent, identically distributed Markov chains in discrete time,
$\{(X^i_n)_{n\ge 0},1\le i\le L\}$, one component for each site.  The
state variable is
\[
X^i_n=\mbox{\# of individuals in generation
$n$ with the derived allele in site $i$}  
\]
and the state space of each chain is $\{0,1,\dots,N\}$.  An example path of
the Markov chain is visualized in  Figure \ref{fig:singlepos}. Site $i$ is
said to be mono-allelic at time $n$ if it carries the ancestral allele
throughout the entire population, so that $X^i_n=0$.  A trajectory
$(X^i_n)_{n\ge 0}$ consists of subsequent mono-allelic periods in state $0$
and active polymorphic periods with both ancestral and derived
alleles present. Whenever a derived allele reaches fixation in
generation $n$, that is $X_n^i=N$, then the derived is declared to be
the new ancestral allele at that site.

We let $\mu>0$ be the mutation
probability
\begin{align*}
\mu = & \textrm{ probability per individual and generation that an 
  ancestral is replaced}\\
&\textrm{ by the derived allele at a single mono-allelic site},
\end{align*}
and for each generation $n$ and site $i$,  we let $J^i_n$ be
binomially distributed independent random variables, such that
for $i=1,\dots,L$, $n\ge 1$,
\[
J_n^i =\mbox{\# of mutations in generation $n$ hitting a mono-allelic site $i$}
           \in \mathrm{Bin}(N,\mu). 
\]
In the limit of small mutation rate $\mu\to 0$, such that $N\mu$ is a
small probability, we have
\[
P(J_n^i=0)=(1-\mu)^N=1-N\mu+o(N\mu)
\]
as well as 
\[
   P(J_n^i=1)=N \mu+o(N\mu),\quad P(J_N^i\ge 2)=o(N\mu).
\]  
Hence, given $X_n^i$ in generation $n$, the
random variable
\[
J^i_{n+1} 1_{\{X^i_n=0\}}
=\mbox{\# of mutations in site $i$ at generation $n+1$}
\]
is approximately $\mathrm{Bin}(1,N\mu)$ distributed, for each $i$. 
It is the injection of new derived alleles in the population at mono-allelic
sites, and the change-of-state of the Markov chain from $0$ to $1$,
which marks the beginning of the active periods. To make the dynamics
during active periods precise, we let $s$, $s\ge -1$, denote the
coefficient of selection. Then, conditionally given $X_n^i=k$
derived alleles at site $i$ in the parental generation $n$, the number
of offspring derived alleles for the next generation $n+1$ is 
a binomially distributed random variable $H^i_{n+1}$, such that
\[
H_{n+1}^i\in \mathrm{Bin}(N,p_k), \quad
p_k=\frac{k(1+s)}{k(1+s)+(N-k)\cdot 1}.
\]
Here, the case $k=N$ represents fixation of the derived allele
and hence the substitution of a former ancestral type with a derived
in site $i$.  In our context, however, the derived is redefined to be
the new ancestral type from generation $n+1$ and onwards, and therefore
the offspring $H^i_{n+1}=N$ will not count towards
$X_{n+1}^i$.  Summing up, given an initial distribution of
$X_0=(X_0^1,\dots,X_0^L)$,  the components of the discrete time Markov chain
$X_n=(X^1_n,\dots,X_n^L)$, are defined recursively by
\[
X_{n+1}^i=H_{n+1}^i-N 1_{\{X_n^i=N\}}
+J^i_{n+1} 1_{\{X^i_n=0\}},\quad  n\ge 0,\quad 1\le i\le L.
\]
It follows that a mono-allelic site remains mono-allelic for a geometrically
distributed number of generations until a single mutation hits after
an average number of $1/N\mu$ generations.
\begin{figure}
\hskip -10mm\includegraphics[width=1.2\textwidth]{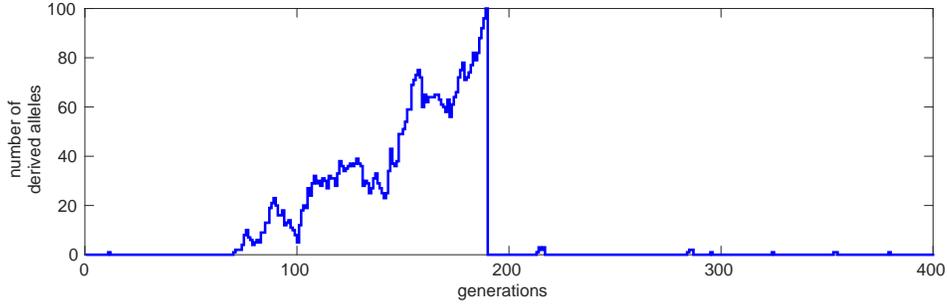}
\caption{Simulation of a single path $\{X_n,0\le n\le 400\}$ for the
  case $N=100$, $s=0$, and $\mu=2\cdot 10^{-4}$. A site is referred to
  as mono-allelic if $X_n=0$. Once a mutation hits a site, the site
  becomes polymorphic $0<X_n<N$, i.e. both the ancestral and the
  derived allele are segregating in the population. The ultimate fate
  of the derived allele is either extinction, $X_n=0$, or fixation,
  $X_n=N$. However, whenever a derived allele reaches fixation in
  generation $n$, then the derived allele is declared to be the new
  ancestral allele at that site, in other words $X_n$ is set to $0$.}
\label{fig:singlepos}
\end{figure}

\subsection{Diffusion approximation}

For the transition to a continuous time Markov chain we introduce the scaled
parameters $\theta>0$ and $\gamma$ defined by
\begin{align*} 
\theta &=L N\mu= \mbox{mutation intensity per sequence and generation}\\
\gamma &=Ns= \mbox{selection intensity per site and generation}.
\end{align*}
Along the relevant evolutionary time scale, $t$ units of time
correspond to $[Nt]$ generations. On this scale the total mutation
rate is $N\theta$ per sequence and time unit.
We define the time-scaled allele frequencies
$X_N(t)=(X_N^1(t),\dots,X_N^L(t))$ by
\[
X_N^i(t)=N^{-1} X_{[Nt]}^i,\quad 1\le i\le L,\quad t\ge 0.
\]
Then, with $h\sim 1/N$ denoting a small evolutionary time step, 
\[
X^i_N(t+h)-X^i_N(t)=\frac{1}{N}H^i_{[Nt]+1}
-1_{\{X^i_N(t)=1\}}-X^i_N(t)
+\frac{1}{N}J^i_{[Nt]+1}1_{\{X^i_N(t)=0\}}.
\]
Furthermore, by evaluating conditional expectations for each term, for
$x\in \{0,1/N,\dots,1\}$,
\begin{align*}
\E[X^i_N(t+h)-X^i_N(t)|X^i_N(t)=x]
& = h\Big(\frac{\gamma x(1-x)}{1+x\gamma/N}\, 1_{\{1/N\le x<1\}}
+\mu N\,1_{\{x=0\}}\Big)\\
& \approx h\Big( \gamma x(1-x)\, 
+\frac{\theta}{L}\,1_{\{x=0\}}\Big),
\end{align*}
where the approximation in the last step comes from ignoring the term
multiplied by $\gamma/N$.  Similarly, by computing second moments,
\begin{align*}
\E[(X^i_N(t+h)-X^i_N(t))^2|X^i_N(t)=x]
&=h\Big( x(1-x)+ \frac{\theta}{NL}\,1_{\{x=0\}}\Big)\\
&\approx h\, x(1-x). 
\end{align*}
The above relations of first and second moments are the approximative
drift and variance functions for the Wright-Fisher diffusion process
with selection, and with the additional mechanism of returns from
state $0$ to state $1/N$ with intensity $N\theta/L$ per site.  The
Wright-Fisher diffusion process arises in the limit of weak convergence of the
Wright-Fisher reproduction model as the population size $N$ tends to
infinity.  Letting $(W_t)_{t\ge 0}$ be a standard Brownian motion, the
Wright-Fisher diffusion process with scaled selection coefficient
$\gamma$, is the Markov process $(\xi_t)_{t\ge 0}$ with state space
$[0,1]$ defined as the unique strong solution of the stochastic
differential equation
\[
d\xi_t=\beta(\xi_t)\,dt+\sigma(\xi_t)\,dW_t,\quad t\ge 0,\quad
\xi_0=x\in (0,1), 
\]   
with drift function $\beta(x)$ and variance function $\sigma^2(x)$ given by
\[
\beta(x)= \gamma x(1-x),\quad \sigma^2(x)=x(1-x).
\]
The case $\gamma=0$ is the neutral Wright-Fisher diffusion process,
$\gamma>0$ corresponds to positive selection and $\gamma<0$ to
negative selection.  Formally, we assume that the paths $(\xi_t)$ are
elements in the class $\Dc$ of functions defined on the real line $\R$
with values in the unit interval $[0,1]$, which are continuous from
the right and have limits from the left.  We write $\Pb^\gamma_x$ for
the probability measure and $\E^\gamma_x$ for the expectation of the
process in $\Dc$, given that $\xi_0=x$.  It is convenient in the
current setting to consider in addition time-shifted processes,
initiated at an arbitrary time $s$.  For such an $s$, let
$\Pb_x^\gamma(d\xi^s)$ be the law of the Wright-Fisher diffusion
process with selection coefficient $\gamma$ and paths $(\xi^s_u)$ in
$\Dc$ with initial time $s$ and initial value $x$, that is
$\xi_u^s=0$, $u<s$, and $\xi_s^s=x$.  This is the strong solution of the
stochastic differential equation
\[
d\xi^s_t=\gamma\xi^s_t(1-\xi^s_t)\,dt+\sqrt{\xi^s_t(1-\xi^s_t)}\,dW_t,
\quad t\ge s, \quad \xi^s_s=x.
\]   
We use the same notations $\Pb_x^\gamma$ and $\E_x^\gamma$ for the
probability measure and expectation without explicit mentioning of the
initial time $s$, which will be clear from context.  With initial
value $x$, $0<x<1$, the process either gets fixed in $x=1$ or goes
extinct in $x=0$ with the corresponding fixation time $\tau_1$,
extinction time $\tau_0$, and absorption time $\tau=\tau_0\wedge
\tau_1$ the minimum of $\tau_0$ and $\tau_1$.  In this sense both
points $\{0,1\}$ are classified as boundary exit points. The exit
measure is given by the fixation probability
\begin{equation} \label{eq:fixprob}
q_\gamma(x)=\Pb_x^\gamma(\tau_1<\infty)=\frac{1-e^{-2\gamma
x}}{1-e^{-2\gamma}},\quad \gamma\not=0,\qquad 
q_0(x)=\Pb_x^0(\tau_1<\infty)=x.
\end{equation}

Based on these observations we now introduce a continuous time Markov
process $Y_{N,L}(t)=(Y_{N,L}^1(t),\dots, Y_{N,L}^L(t))$, which is our
final population model for the case of large but finite $N$ and fixed
$L$. An example of such a process is visualized in Figure \ref{fig:pathcuts}.
The components $Y_{N,L}^i(t)$ of $Y_{N,L}$ have state space
given by the continuous interval $[0,1)$ and jumps from the boundary.
The paths are cyclic with each cycle consisting of one mono-allelic period
in state $0$ and one polymorphic period of non-zero frequency. The
mono-allelic periods are exponentially distributed with intensity
$N\theta/L$. During an active period, starting at time $s$, the path
is a Wright-Fisher diffusion process $(\xi^s_t)_{t\ge s}$ with initial
state $1/N$.  The duration of the active period is the absorption time
$\tau$. The result is either fixation, which occurs with the scaled
fixation probability
\begin{equation}\label{def:omega}
\Pb_{1/N}^\gamma(\tau_1<\infty)\approx \frac{\omega_\gamma}{N},\quad 
\omega_\gamma=\lim_{N\to\infty}Nq_\gamma(1/N)=\frac{2\gamma}{1-e^{-2\gamma}},
\quad \omega_0=1,
\end{equation}
or, otherwise, extinction. The end of the active period marks the
beginning of a new mono-allelic period, hence a new cycle.
The previously studied frequency processes $X^i_N(t)$, are
discrete state approximations of $Y_{N,L}^i(t)$. The collection of allele
frequencies $Y_{N,L}^i(t)>0$ over all $L$ sites constitutes the AFS
at time $t$. Three example AFS are depicted by the histograms shown in the lower panel of Figure \ref{fig:pathcuts}.
\begin{figure}[t] 
\hskip -14mm \includegraphics[width=1.2\textwidth]{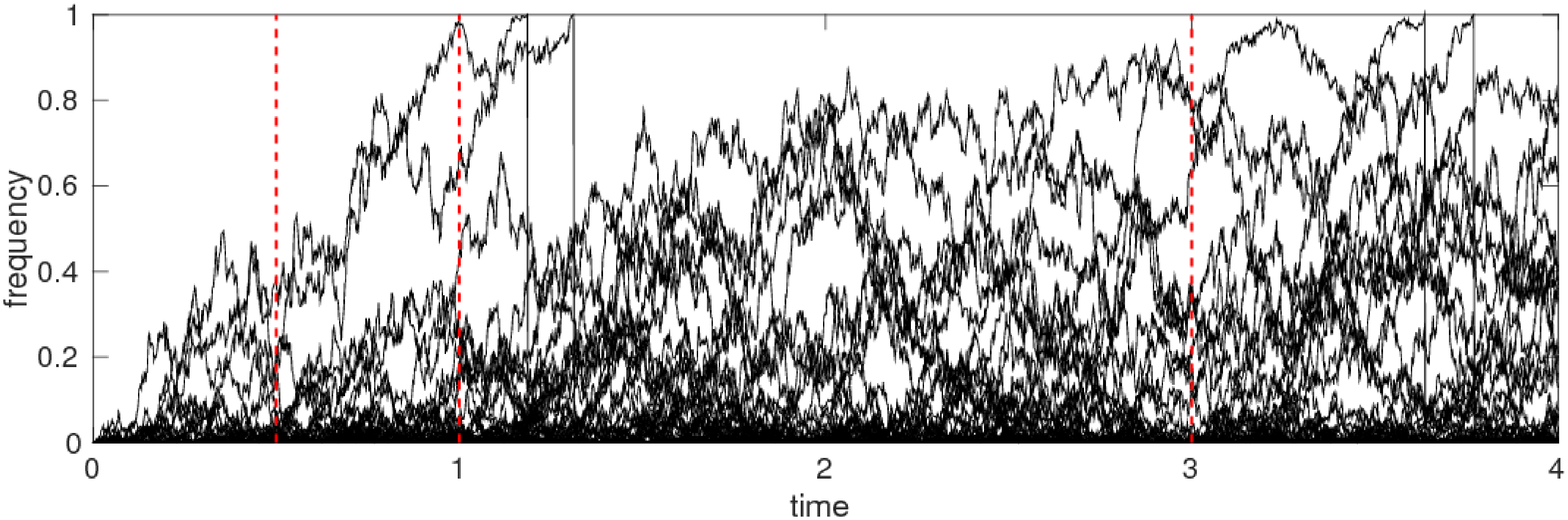}

\hskip -14mm \includegraphics[width=1.2\textwidth]{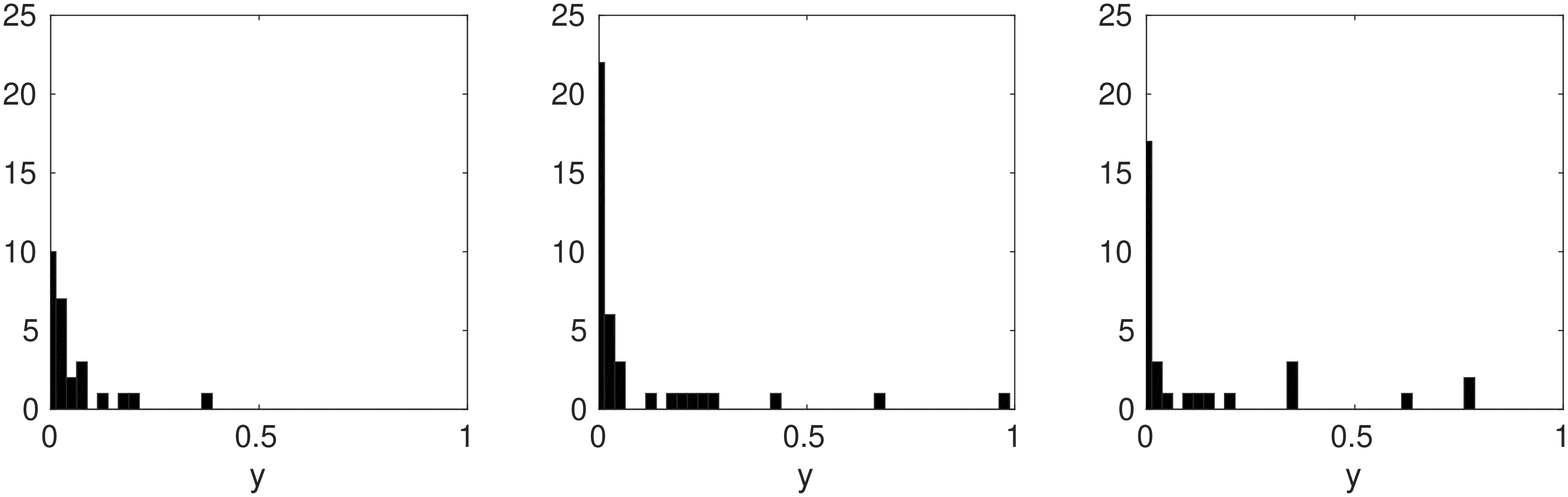}

\caption{Upper panel: a simulation of $Y_{N,L}(t)$, $0\le t\le 4$, for a
  population of size $N=500$ consisting of $L=1500$ sites, which are
  all mono-allelic at $t=0$ and evolve neutrally, $\gamma=0$, with a
  mutation rate $\theta=3$. Each black line represents the path of the
  derived allele frequency at one of the $L=1500$ sites. The three vertical
  red dashed lines mark the time points $t=0.5,1,3$, for which the AFS is
  depicted by histograms. Lower panel: the three histograms display the
  AFS for the time points $t=0.5,1,3$ from left to right, which illustrates the
  build-up of higher allele frequencies with time. }
\label{fig:pathcuts}
\end{figure}

\subsection{Poisson random field approximation}

The next stage in developing the model is concerned with calibrating
the length of the sequences measured in sites, $L$, with the strength
of mutation per site, $N\theta/L$.  The basic observation is that the
number of new mutations per time unit in $L$ mono-allelic sites is
$\mathrm{Bin}(L,N\theta/L)$, hence for large $L$ but fixed population
size $N$, approximately Poisson with mean $N\theta$.  Rather than
visualizing infinitely long sequences it is convenient therefore to imagine
a spatially continuous sequence (of length one, say), where polymorphic
``Poisson sites'' are placed according to a Poisson process with
intensity $N\theta$. 

The collection of allele frequencies over all $L$ sites $\{Y_{N,L}(t),t\ge 0\}$ form a
random field $\Xc_t^{N,L}$ on the positive real line, in the sense
\[
t\mapsto \Xc_t^{N,L}=\sum_{i=1}^L \delta_{Y_{N,L}^i(t)},\quad t\ge 0.
\]
For bounded functions $f:[0,1]\to \R$ on the unit interval $[0,1]$ we
use the bracket notation $\langle\Xc_t^{N,L}\!,f\rangle$ for the
application of the random field to $f$, and evaluate the combined
effect of all sites by
\[
t\mapsto \langle\Xc_t^{N,L}\!,f\rangle 
=\sum_{i=1}^L f(Y_{N,L}^i(t)). 
\]
To handle the limit operations as $L$ and, later, $N$ tend to infinity,
we specify the set of functions
\[
\Fc=\{f: [0,1]\to\R,\; \mbox{$f$ bounded},\; f(0)=0,\; \int_0^1
y^{-1}|f(y)|\,dy<\infty\},
\]
and the further restricted subset
\[
\Fc_0=\{f\in \Fc: f(1)=0\}.
\]
The class $\Fc_0$ enters naturally studying properties of the
Wright-Fisher diffusion process based on the Green function and
occupation time functional for diffusion processes,
\citep{KarlinTaylor1981,Breiman1992}.

The scale function $S_\gamma(x)$ and speed function $m_\gamma(x)$
associated with the Wright-Fisher diffusion process with selection parameter
$\gamma$ are 
\[ 
S_0(x)=x,\quad  S_\gamma(x)=\frac{1}{2\gamma}(1-e^{-2\gamma
  x}),\quad \gamma\not=0,\qquad 
m_\gamma(x)=\frac{e^{2\gamma x}}{x(1-x)}.
\]
By integration with respect to the Green function $G_\gamma(x,y)$, defined as
\[
G_\gamma(x,y)=\left\{
\begin{array}{lc} 
2q_\gamma(x)(S_\gamma(1)-S_\gamma(y))m(y), & 0\le x\le y\le 1 \\[2mm]
2(1-q_\gamma(x))(S_\gamma(y)-S_\gamma(0))m(y), & 0\le y\le x\le 1,
\end{array}
\right. 
\]
one obtains the time occupation functional 
\begin{equation}\label{eq:occupationtime}
\E^\gamma_x\Big[\int_0^{\tau}g(\xi_t)\,dt\Big]
=\int_0^1 G_\gamma(x,y)g(y)\,dy. 
\end{equation}
Using the functions
\[
\pi_\gamma(y)=\frac{1-e^{-2\gamma(1-y)}}{\gamma y(1-y)},\quad 
\widetilde \pi_\gamma(y)=\frac{e^{2\gamma y}-1}{\gamma y(1-y)},\quad
\gamma\not=0
\]
and
\[
\pi_0(y)=\frac{2}{y}\quad \widetilde \pi_0(y)=\frac{2}{1-y}
\]
we have 
\[ 
G_\gamma(x,y)=q_\gamma(x)\,1_{\{x<y\}}\pi_\gamma(y)+
(1-q_\gamma(x))\,1_{\{x>y\}}\widetilde\pi_\gamma(y)
\] 
and
\begin{equation}\label{eq:intgreen}
\int_0^1 G_\gamma(x,y) g(y)\,dy
=q_\gamma(x)\int_x^1 g(y)\pi_\gamma(y)\,dy
 +(1-q_\gamma(x))
\int_0^x g(y)\widetilde\pi_\gamma(y)\,dy
\end{equation}
whenever the integrals on the right hand side are well-defined.  It is
now straightforward to derive from (\ref{eq:occupationtime}) and
(\ref{eq:intgreen}), and also using the parameter
$\omega_\gamma$ introduced in (\ref{def:omega}), the following
well-known and fundamental limit property of Wright-Fisher diffusion processes: 
\begin{equation}\label{eq:scaledoccupationtime}
\lim_{N\to\infty}N\E_{1/N}^\gamma\int_0^\tau f(\xi_s)\,ds
=\omega_\gamma\int_0^1 f(y)\pi_\gamma(y)\,dy, \quad f\in\Fc_0,
\end{equation}
where $\omega_\gamma\, \pi_\gamma(y)$ is known as the allele frequency spectrum.  
 
We are now in position to introduce a random field $\Xc_t^N$, which
arises from $\Xc_t^{N,L}$ in the limit $L\to\infty$.  We first 
construct $\Xc_t^N$ as a stochastic integral with respect to a
Poisson measure, see e.g.\ \citep{Kallenberg2002}, and then establish the
convergence in $L$.  Let $n_N(ds,d\xi^s)$ be the product measure
defined on $\R^+\times \Dc$ by
\[
n_N(ds,d\xi^s)=N\theta\,ds\, \Pb_{1/N}^\gamma(d\xi^s)
\]
and let $\Nc_N(ds,d\xi^s)$ be a Poisson random measure on $\R^+\times
\Dc$ with intensity measure given by $n_N$.  
For $t\ge 0$, we let $\Xc_t^N$ be the Poisson random field  
defined by the stochastic integral
\[
\langle \Xc_t^N\!,f\rangle =\int_{\R^+\times\Dc}
f(\xi^s_t)\,\Nc_N(ds,d\xi^s).
\]

\begin{proposition} \label{prop:stochint}
The stochastic integral $\langle\Xc^N_t\!,f\rangle$ is well-defined
with finite expected value, such that, 
for every $f\in \Fc_0$,  
\[
\E \langle\Xc^N_t\!,f\rangle
= N\theta\, \E^\gamma_{1/N}\Big[\int_0^{t\wedge \tau} f(\xi^0_u)
  \,du\Big]<\infty
\]
and, for $f\in \Fc$ and fixed $t$, 
\[
\E \langle\Xc^N_t\!,f\rangle
= N\theta \,\E^\gamma_{1/N}\Big[\int_0^t f(\xi^0_u)\,du\Big]<\infty.
\]
\end{proposition}
\begin{proof} 
For the existence and finite expectation of the Poisson integral 
it is sufficient to show
\[
\int_{\R^+\times\Dc}|f(\xi^s_t)|\,n_N(ds,d\xi^s)
=N\theta \E^\gamma_{1/N}\int_0^{t\wedge\tau} |f(\xi_t^s)|\,ds <\infty.
\]
However, for $f\in\Fc_0$ the right hand side is bounded by 
\[
N\theta \E^\gamma_{1/N}\int_0^\tau |f(\xi_t^s)|\,ds
=N\theta \int_0^1 G_\gamma(1/N,y)|f(y)|\,dy,
\]
which is finite for every $N$ and, moreover, has a finite limit as
$N\to\infty$.  Thus, $\langle\Xc^N_t\!,f\rangle$ exists with finite
expected value
\[
N\theta \E^\gamma_{1/N}\int_0^{t\wedge\tau} f(\xi^s_t)\,ds
=N\theta \E^\gamma_{1/N}\int_0^{t\wedge\tau} f(\xi^0_u)\,du.
\]
A similar argument shows the claim for $f\in \Fc$ and fixed $t$. 
\end{proof}

\begin{proposition} \label{prop:Ltoinfty} Let $f\in \Fc_0$. 
The allele frequency random field $\langle\Xc_t^{N,L}\!,f\rangle$ converges
as $L$ tends to infinity to the Poisson random field
$\langle\Xc_t^N\!,f\rangle$, 
\[
\{\langle\Xc_t^{N,L}\!,f\rangle,\;t\ge 0\} 
\Rightarrow \{\langle\Xc_t^N\!,f\rangle,\; t\ge 0\},
\]
in the sense of convergence of random processes in finite-dimensional
distribution.  
\end{proposition}

The proof of this proposition employs an indexing method using
measures, which amounts to taking the limit as $L\to\infty$ of the
quantities
\[
 \int \langle\Xc_u^{N,L}\!,f\rangle\,\mu(du),
\]
for a suitable class of measures $\mu$.  This technique has been used
for other random fields elsewhere, and the technical aspects are not
central for the specific problem at hand. Therefore, the proof is
given separately in Section \ref{sec:technical}.

\subsection{The stationary functional}

In our approach, we start from a completely mono-allelic state at time $t=0$,
where $\langle\Xc_t^N\!,f\rangle$, $t\ge 0$, represents the non-equilibrium build-up
of allele frequencies towards a steady-state spectrum of frequencies in the
limit $t\to\infty$. One way of identifying such an asymptotic limit is achieved
by tuning the initial state of the process to obtain a stationary system. Here,
the appropriate initial state is obtained simply by including all mutations
which occurred at some time $s<0$ and counting the resulting frequencies
$\xi_t^s$ at time $t>0$. The result is a stationary version $\langle\Xc_\infty^N,f\rangle$
of the population functional $\langle\Xc_t^N\!,f\rangle$. Indeed,
observing that $f\in \Fc_0$ implies $g=e^{\alpha f}-1\in\Fc_0$,
\[
\langle\Xc_t^N\!,f\rangle=
\int_{\R_+\times\Dc} f(\xi^s_t)\,\Nc_N(ds,d\xi^s), \quad f\in \Fc_0, 
\]
has the Laplace functional
\begin{align*}
\ln \E \exp\{\alpha \langle\Xc_t^N\!,f\rangle\}
&= N\theta\,\E_{1/N}^\gamma
\Big[\int_0^{t\wedge\tau} (e^{\alpha  f(\xi^s_t)}-1)\,ds\Big]\\
&= N\theta\,\E_{1/N}^\gamma
\Big[\int_0^{t\wedge\tau} (e^{\alpha  f(\xi^0_u)}-1)\,du\Big],
\end{align*}
which converges as $t\to\infty$ to
\[
N\theta\,\E_{1/N}^\gamma
\Big[\int_0^\tau(e^{\alpha  f(\xi^0_u)}-1)\,du\Big]
=\ln \E \exp\{\alpha \langle\Xc_\infty^N,f\rangle\}.
\]
Here, the intensity measure of 
\[
\langle\Xc_\infty^N,f\rangle=
\int_{\R\times\Dc} f(\xi^s_t)\,\Nc_N(ds,d\xi^s)
\]
in the variable $s$ extends to all of the real line. Furthermore, by
(\ref{eq:scaledoccupationtime}),
\[
\lim_{N\to\infty} \ln \E \exp\{\alpha \langle\Xc_\infty^N,f\rangle\}
=\theta \omega_\gamma \int_0^1 (e^{\alpha f(y)}-1)\pi_\gamma(y)\,dy.
\]
Hence $\langle\Xc_\infty^N,f\rangle$ converges in distribution as
$N\to\infty$ to the stochastic integral 
\begin{equation}\label{eq:statstochint}
\langle\Xc_\infty,f\rangle=\int f(y)\,\Nc(dy),
\end{equation}
where $\Nc(dy)$ is a Poisson random measure on $[0,1]$ with intensity
measure $\theta\omega_\gamma \pi_\gamma(y)dy$.

\section{The non-equilibrium allele frequency spectrum}

The limiting Poisson intensity $\omega_\gamma \pi_\gamma(y)dy$ of the
stationary model $\Xc_\infty$ in (\ref{eq:statstochint}) appeared in
(\ref{eq:scaledoccupationtime}) as the kernel of a scaled occupation time
functional. The main theoretical result in this work is the
derivation of the fixed time and large population size limit of the
Poisson integral expectations
\[
\lim_{N\to\infty} \E^\gamma\langle \Xc_t^N\!,f\rangle
\] 
in Proposition \ref{prop:stochint}.  In doing so we obtain a
non-equilibrium version of the AFS which represents the build-up
of frequencies over a time period $[0,t]$. To simplify notation from now on we write
$(\xi_t)_{t\ge 0}$ (rather than $(\xi^0_t)$) for the the Wright-Fisher
diffusion process with initial time $t=0$. Then
\[
\E\langle \Xc_t^N\!,f\rangle = 
\theta N\E_{1/N}^\gamma\int_0^t f(\xi_u)\,du,\quad f\in \Fc
\]
and
\[ 
N\E_{1/N}^\gamma\int_0^t f(\xi_u)\,du
=N \E_{1/N}^\gamma\int_0^{t\wedge\tau} f(\xi_u)\,du,\quad f\in \Fc_0.
\]

\subsection{Representation in terms of the probability distribution of the time to fixation}

Let $\Pb^{*\gamma}_x$ and $\E^{*\gamma}_x$ be the distribution and
expectation of the Wright-Fisher diffusion process with selection
coefficient $\gamma$ conditioned on the event of fixation,
$\tau_1<\infty$.  Then, by symmetry,
\begin{align*}
\Pb^\gamma_x(\tau< t)&=\Pb^\gamma_x(\tau_1<
t)+\Pb^\gamma_x(\tau_0< t)\\
&=\Pb^{*\gamma}_x(\tau_1< t)q_\gamma(x)
+\Pb^{*-\gamma}_{1-x}(\tau_1< t)(1-q_\gamma(x)),
\end{align*}
where $q_\gamma(x)$ is the fixation probability defined in
(\ref{eq:fixprob}).  For the neutral case $\gamma=0$, the fixation
time distribution given that fixation occurs is given by 
\citep{Kimura1970b}
\begin{equation}\label{eq:kimurafixation}
\Pb_x^{*0}(\tau_1>t)=(1-x)\sum_{i=2}^\infty (2i-1)(-1)^i
H([2-i,i+1],[2],x)\,e^{-\binom{i}{2}t}, 
\end{equation}
where $H$ is the hypergeometric function. In particular,
\[ 
\Pb_0^{*0}(\tau_1>t)=\sum_{i=2}^\infty (2i-1)(-1)^i \,e^{-\binom{i}{2}t}.
\]
The time-dependent integration kernel of the non-equilibrium AFS
turns out to be
$\omega_\gamma\Pb^{*\gamma}_{1-y}(\tau_1\le t) \,\pi_\gamma(y)dy$, in
the sense of the following result. 

\begin{theorem}\label{thm:main} Let $f\in \Fc_0$.  Then 
\begin{align}\label{eq:thmmainlimit}
 \lim_{N\to\infty} \E\langle \Xc_t^N\!,f\rangle
&=\theta\omega_\gamma\int_0^1 f(y) \,\Pb^{*\gamma}_{1-y}(\tau_1\le t)
\pi_\gamma(y)\,dy\\
\label{eq:thmmainlimit2}
&=\theta\int_0^1 f(y) \frac{\Pb^{\gamma}_{1-y}(\tau_1\le t)}{1-y}\pi_0(y)\,dy
\end{align}
More generally, if $f\in \Fc$, then 
\begin{align*}
\lim_{N\to\infty} \E\langle \Xc_t^N\!,f\rangle 
= \theta\omega_\gamma \int_0^1 (f(y)-f(1) q_\gamma(y))\,
\Pb_{1-y}^{*\gamma}(\tau_1\le t)\pi_\gamma(y)\,dy+\theta f(1)\omega_\gamma\, t
\end{align*}
and
\begin{align}\nonumber
\lim_{N\to\infty} 
N\E_{1/N}^\gamma \int_0^{t\wedge\tau} f(\xi_u)\,du 
&= \theta\omega_\gamma \int_0^1 (f(y)-f(1) q_\gamma(y))\,
\Pb_{1-y}^{*\gamma}(\tau_1\le t)\pi_\gamma(y)\,dy\\
&\quad +\theta f(1)\omega_\gamma \,\int_0^t \Pb^{*\gamma}_0(\tau_1>s)\,ds.
\label{eq:thmmainlimit3}
\end{align}
\end{theorem}

\begin{proof} The equivalence of (\ref{eq:thmmainlimit}) and
  (\ref{eq:thmmainlimit2}) is immediate from 
\[
\Pb^{*\gamma}_{1-y}(\tau_1\le t)
=\frac{\Pb^\gamma_{1-y}(\tau_1\le t)}{q_\gamma(1-y)},
\]
hence
\[
\omega_\gamma \Pb^{*\gamma}_{1-y}(\tau_1\le t) \pi_\gamma(y)
=\frac{\Pb^\gamma_{1-y}(\tau_1\le t)}{1-y}\, \pi_0(y).
\]
To prove (\ref{eq:thmmainlimit}) we apply a time reversal technique.
Let $\nu_N$ be the two-state distribution which gives probability
$q_\gamma(1/N)$ to $1/N$ and $1-q_\gamma(1/N)$ to $1-1/N$ and assume
that $\{\eta_v,v\ge 0\}$ with distribution $\Pb_{\nu_N}^{*\gamma}$ is
a Wright-Fisher diffusion process with initial distribution $\nu_N$
and selection coefficient $\gamma$, which is conditioned on ultimate
fixation in state $1$.  The essence in the construction of $\eta$ is
to provide a close mimic of the time and space reversal of $\xi$
defined by $1-\xi_{\tau-v}$, $0\le v\le \tau$.  Next, take $f\in
\Fc_0$ and observe that $\widetilde \xi_t=1-\xi_t$ defines a
Wright-Fisher diffusion with selection coefficient $-\gamma$, initial
value $\widetilde \xi_0=1-1/N$, and the same absorption time as
$(\xi_t)$.  Hence
\[
\E_{1/N}^\gamma\int_0^\tau f(\xi_u)\,1_{\{u\le t\}}\,du
=\E_{1-1/N}^{-\gamma}\int_0^\tau f(1-\widetilde \xi_u)\,1_{\{u\le t\}}\,du.
\]
Reversing time, 
\[
\lim_{N\to\infty}\E^{-\gamma}_{1-1/N}
\int_0^\tau f(1-\widetilde \xi_u)\,1_{\{u\le t\}}\,du
=\lim_{N\to\infty}\E^{*\gamma}_{\nu_N}\int_0^{\tau_1} 
f(1-\eta_v)\,1_{\{\tau_1\le v+t\}}\,dv.
\]
Let $\Fc_v=\sigma(\eta_u,u\le v)$ be the minimal $\sigma$-algebra
generated by $(\eta_v)$. Then $\tau_1$ is an $(\Fc_v)$-stopping time
so, by conditioning,
\begin{align*}
\E^{*\gamma}_{\nu_N}\int_0^{\tau_1}
f(1-\eta_v)\,1_{\{\tau_1\le v+t\}}\,dv
=\E^{*\gamma}_{\nu_N}\int_0^{\tau_1} 
f(1-\eta_v)\,\E^{*\gamma}[1_{\{\tau_1\le v+t\}}|\Fc_v]\,dv
\end{align*}
Here, shifting $\tau_1$ to $\widetilde \tau_1=\tau_1-v$, with the same
distribution, we have  
\[
\E^{*\gamma}[1_{\{\tau_1\le v+t\}}|\Fc_v]=
   \E^{*\gamma}_{\eta_v}[1_{\{\widetilde\tau_1\le t\}}]
  = \Pb^{*\gamma}_{\eta_v}(\tau_1\le t)
\]
so that
\begin{align*}
\E^{*\gamma}_{\nu_N}\int_0^{\tau_1}
f(1-\eta_v)\,1_{\{\tau_1\le v+t\}}\,dv
=\E^{*\gamma}_{\nu_N}\int_0^{\tau_1} 
f(1-\eta_v)\,\Pb^{*\gamma}_{\eta_v}(\tau_1\le t)\,dv.
\end{align*}
Writing $g(x)=f(x)\,\Pb^{*\gamma}_{1-x}(\tau_1\le t)$, the integrand
on the right hand side in the previous expression takes the form
$g(1-\eta_v)$.  Hence, to complete the proof of (\ref{thm:main}), it
remains to show that, for any $g\in \Fc_0$,
\begin{equation}\label{eq:thmkeylimit} 
\lim_{N\to\infty}N\E^{*\gamma}_{\nu_N}\int_0^{\tau_1} g(1-\eta_v)\,dv
=\omega_\gamma\int_0^1 g(y)\pi_\gamma(y)\,dy.
\end{equation}
However, the representation of the time occupation functional in
(\ref{eq:scaledoccupationtime}) extends to the measure conditional on
fixation \citep{KarlinTaylor1981}. Namely,
\[
\E^{*\gamma}_x\int_0^{\tau_1} g(\xi_t)\,dt
=\int G_*(x,y)g(y)\,dy,
\]
where
\[
G_*(x,y)=\left\{
\begin{array}{lc} 
2S_\gamma(y)(S_\gamma(1)-S_\gamma(x)) \displaystyle{\frac{S_\gamma(y)m_\gamma(y)}{S_\gamma(1)S_\gamma(x)}}, & 0\le y\le x\le 1\\[4mm]
2S_\gamma(x)(S_\gamma(1)-S_\gamma(y)) \displaystyle{\frac{S_\gamma(y)m_\gamma(y)}{S_\gamma(1)S_\gamma(x)}}, & 0\le x\le
y\le 1. 
\end{array}
\right. 
\]
for the same scale function $S_\gamma$ and and speed function
$m_\gamma$ as in (\ref{eq:occupationtime}). Thus,
\begin{align*}
\E^{*\gamma}_{\nu_N}\int_0^{\tau_1} g(1-\eta_v)\,dv
&= q_\gamma(1/N) \int_0^1 
G_*(1/N,y)\,g(1-y)\,dy\\  
&\quad +(1-q_\gamma(1/N))\int_0^1 G_*(1-1/N,y)\, g(1-y)\,dy,  
\end{align*}
and so 
\begin{align}\nonumber
N\E^{*\gamma}_{\nu_N}\int_0^{\tau_1} g(1-\eta_v)\,dv
&\sim  \omega_\gamma\int_0^1 G_*(1/N,y)\,g(1-y)\,dy\\  
&\quad +N\int_0^1 G_*(1-1/N,y)\, g(1-y)\,dy.
\label{eq:thmmainsumup}
\end{align}
Here,  asymptotically as $N\to\infty$ 
\begin{align*}
 \int_0^1 G_*(1/N,y)&g(1-y)\,dy\\
& =\frac{(e^{-2\gamma/N}-e^{-2\gamma})}{\gamma(1-e^{-2\gamma/N})(1-e^{-2\gamma})}
\int_0^{1/N} \frac{(1-e^{-2\gamma y})^2g(1-y)}{y(1-y)e^{-2\gamma
    y}}\,dy\\
&\quad +\int_{1/N}^1 \frac{1-e^{-2\gamma(1-y)}}{\gamma y(1-y)} \frac{1-e^{-2\gamma y}}
{1-e^{-2\gamma}}g(1-y)\,dy\\
&\sim \frac{g(1)}{N}+\int_0^1 \frac{1-e^{-2\gamma(1-y)}}{\gamma y(1-y)} 
\frac{1-e^{-2\gamma y}}{1-e^{-2\gamma}}g(y)\,dy
\end{align*}
and
\begin{align*}
 \int_0^1 G_*(1-&1/N,y)g(1-y)\,dy\\
&= \frac{(e^{-2\gamma(1-1/N)}-e^{-2\gamma})}{\gamma(1-e^{-2\gamma(1-1/N)})(1-e^{-2\gamma})}\int_0^{1-1/N} \frac{(1-e^{-2\gamma y})^2g(1-y)}{y(1-y)e^{-2\gamma y}}\,dy\\
&\quad +\int_{1-1/N}^1 \frac{1-e^{-2\gamma(1-y)}}{\gamma y(1-y)}
\frac{1-e^{-2\gamma y}} {1-e^{-2\gamma}}g(1-y)\,dy\\
&\sim \frac{\omega_\gamma}{N} \int_0^1 \frac{1-e^{-2\gamma(1-y)}}{\gamma y(1-y)} 
\frac{e^{-2\gamma y}-e^{-2\gamma}}{1-e^{-2\gamma}}g(y)\,dy+\int_0^{1/N}g(y)\,dy.
\end{align*}
By adding up the terms in (\ref{eq:thmmainsumup}) and using $g(0)=0$ and
$g(1)$ bounded we obtain
\[
\lim_{N\to\infty}N\E^{*\gamma}_{\nu_N}\int_0^{\tau_1} g(1-\eta_v)\,dv
=\omega_\gamma \int_0^1 g(y)\pi_\gamma(y)\,dy,
\]
as required to verify (\ref{eq:thmkeylimit}) and hence (\ref{eq:thmmainlimit}).
The extension from $\Fc_0$ to $\Fc$ now follows by an application of 
Lemma \ref{lem:extendmain} in Section \ref{sec:technical}.
\end{proof}

The representation of the time-dependent AFS in terms of the
probability distribution of the time to fixation allows for exact
analytical solutions  (Figure \ref{fig:afs}). 
\begin{figure}
\includegraphics[width=0.5\textwidth]{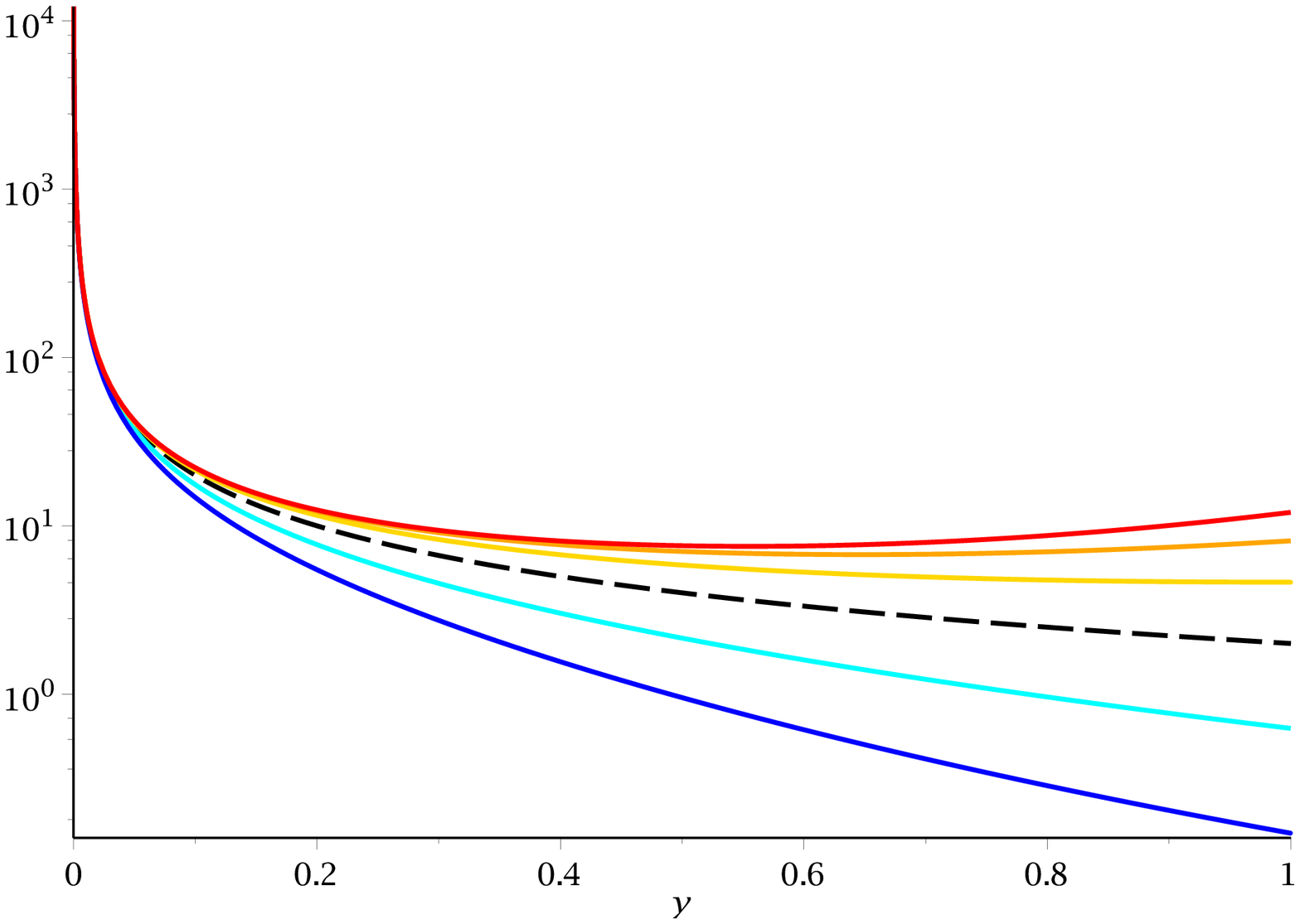}
\includegraphics[width=0.5\textwidth]{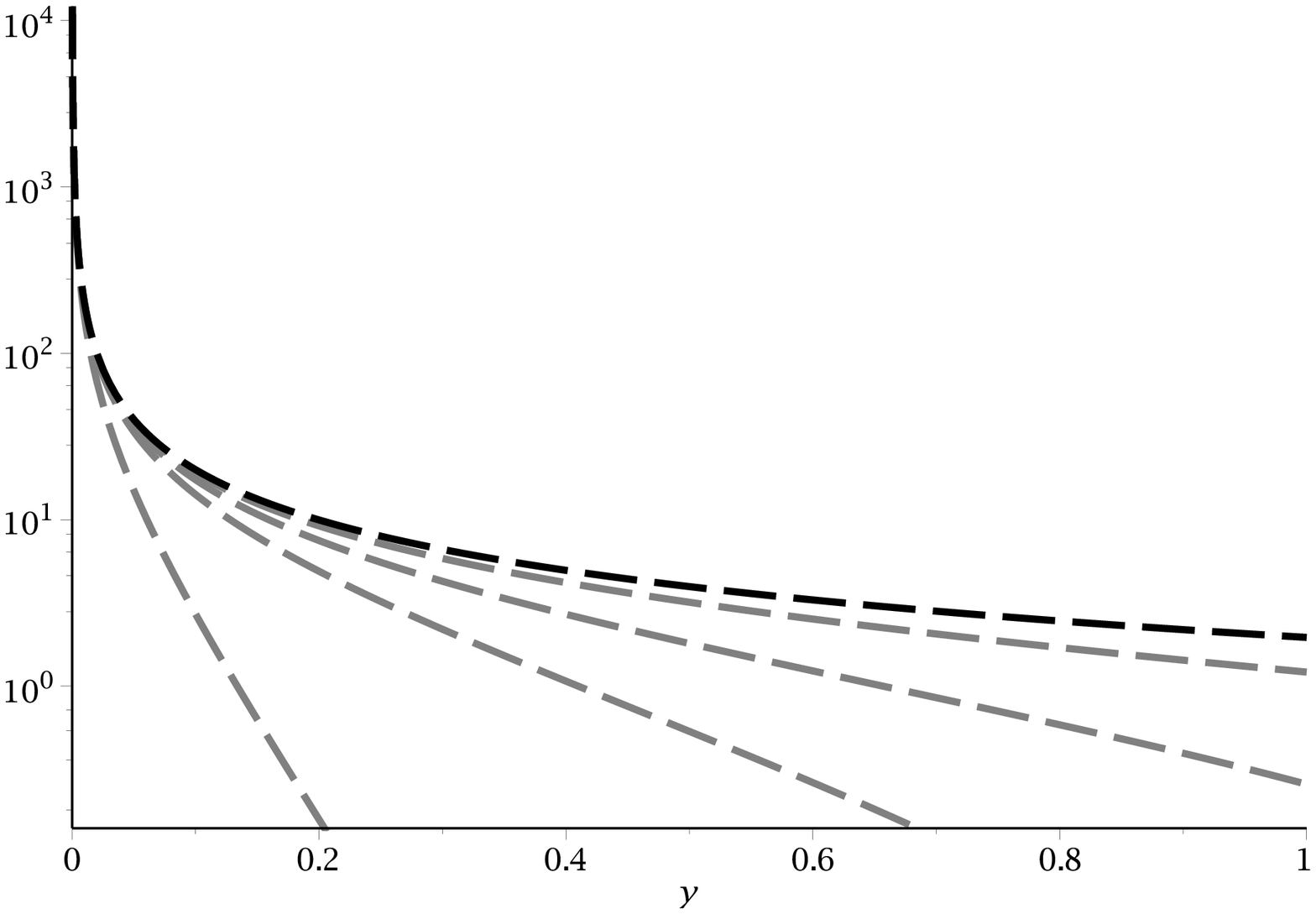}
\caption{Left panel: the stationary AFS
for six scenarios of selection, $\gamma=-2$ (blue solid line), $\gamma=-1$
(turquoise solid line), $\gamma=0$ (black dashed line),  $\gamma=1$
(yellow solid line), $\gamma=2$ (orange solid line) and $\gamma=3$
(red solid line). Right panel:  the build-up of
the AFS starting from a completely mono-allelic state at $t=0$ for the
neutral case, $\gamma=0$. The gray dashed lines represent the AFS for the time
points $t=0.1$, $t=0.5$, $t=1$ and $t=2$ in increasing order. The black dashed
line represents the equilibrium AFS.}
\label{fig:afs}
\end{figure}
To motivate the extension to the case $\Fc$ in Theorem \ref{thm:main},
we mention $f_\mathrm{fix}\in\Fc$ defined by
$f_\mathrm{fix}(y)=1_{\{y=1\}}$. Then
\[
\langle \Xc_t^N\!,f_\mathrm{fix}\rangle=\mbox{\# of 
  alleles that reach fixation in $[0,t]$}
\]
and
\[
\lim_{N\to\infty} \E\langle \Xc_t^N\!,f_\mathrm{fix}\rangle
= -\theta\omega_\gamma \int_0^1  q_\gamma(y)\,
\Pb_{1-y}^{*\gamma}(\tau_1\le t)\pi_\gamma(y)\,dy+\theta \omega_\gamma\, t.
\]
But, clearly, $\int_0^{t\wedge\tau} f_\mathrm{fix}(\xi_u)\,du=0$
and so, by (\ref{eq:thmmainlimit3}),
\[
\lim_{N\to\infty} \E\langle \Xc_t^N\!,f_\mathrm{fix}\rangle
=\theta \omega_\gamma\, t
-\theta \omega_\gamma \,\int_0^t \Pb^{*\gamma}_0(\tau_1> s)\,ds
=\theta\omega_\gamma(t-\E^{*\gamma}_0(\tau_1\wedge t)).
\]
In contrast to the work by \cite{Sawyer1992}, who start from a stationary state,
this relation, provides the precise growth in the number of fixations
starting from a completely mono-allelic population at time $t=0$. Thus, it can,
for example, be applied to study the number of fixations of lineage-specific
mutations after a population split \citep{Mugal2014}.

\subsection{Representation using the duality relationship}

In this subsection, we establish an additional representation of the
limit functional in Theorem \ref{thm:main}, valid for the neutral scenario or
the case of negative selection $\gamma\le 0$. This representation of the
non-equilibrium AFS results by rewriting
(\ref{eq:thmmainlimit}), equivalently (\ref{eq:thmmainlimit2}), using
the duality between the Wright-Fisher diffusion process and a class of
birth-death processes generalizing Kingman's coalescent process.
Kingman's pure death coalescent process \citep{Kingman1982}, is the
Markov process $(A_t)_{t\ge 0}$ defined by jump rates
\begin{equation}\label{eq:kingmanintensity}
\Pb^0(A_{t+h}-A_t=-1|A_t=k)=\binom{k}{2}h+o(h),\quad h\to 0.
\end{equation}
Writing $\Pb_m^0$ and $\E^0_m$ for the conditional law and expectation
given $A_0=m$, Kingman's coalescent $(A_t)$ is a
dual process of the neutral Wright-Fisher diffusion $(\xi_t)_{t\ge 0}$,
in the sense
\[
\E^0_x(\xi_t^m)=\E^0_m(x^{A_t})
\]
\citep{Tavare1984}, and since the distributional properties of $(A_t)$
are known, the duality relation provides an important
computational tool. Indeed, \cite{Griffiths1979} and \cite{Tavare1984}
showed that, under $\Pb_m^0$, 
the Markov transition probabilities of $A_t$ have the representations
\begin{align}\label{eq:distrpassage}
\Pb^0_m(A_t=1)&=1-\sum_{k=2}^m e^{-\binom{k}{2}t}
(2k-1)(-1)^k\frac{m_{[k]}}{m_{(k)}}\\
\Pb^0_m(A_t=j)&=
\sum_{k=j}^m e^{-\binom{k}{2}t}\,\frac{(2k-1)(-1)^{k-j}j_{(k-1)}
}{j!(k-j)!}\frac{m_{[k]}}{m_{(k)}},\quad 2\le j\le m,
\label{eq:distrkingman}
\end{align}
with increasing and decreasing factorials defined as
\[
m_{[k]}=\frac{\Gamma(m+1)}{\Gamma(m-k+1)},\quad
m_{(k)}=\frac{\Gamma(m+k)}{\Gamma(m)}.
\]
Furthermore, 
\begin{equation} \label{eq:expkingman}
\E^0_m(A_t)=1+\sum_{k=2}^m e^{-\binom{k}{2}t}
(2k-1)\frac{m_{[k]}}{m_{(k)}}.
\end{equation}
These formulas remain valid as $m\to\infty$ with the replacement
$m_{[k]}/m_{(k)}\to 1$, for example
\[
\E^0_\infty(A_t)=1+\sum_{k=2}^\infty e^{-\binom{k}{2}t}(2k-1).
\]

The moment duality relation between the neutral Wright-Fisher
diffusion process and Kingman's coalescent process extends to duality
between the Wright-Fisher diffusion process with selection and a wider
class of birth-and-death processes, for which we keep the notation
$A_t$ \citep{Shiga1986,Athreya2005}. For $\gamma\ge 0$,
$(A_t)_{t\ge 0}$ is a birth-death process with linear birth intensity
$\gamma$, such that
\begin{equation}\label{eq:branchingintensity}
\Pb(\eta_{t+h}-\eta_t=1|\eta_t=k)=\gamma kh+o(h),\quad h\to 0,
\end{equation}
and death intensity the same as in (\ref{eq:kingmanintensity}).  Then
$A_t$ possesses a steady-state $A_\infty$, which has the distribution
of a Poisson random variable with mean $2\gamma$, conditioned to stay
positive. The duality relation is 
\[
\E^\gamma_x[(1-\xi_t)^m]=\E^\gamma_m[(1-x)^{A_t}].
\]
By symmetry,
\[
\E^{-\gamma}_x[\xi_t^m] = \E^\gamma_{1-x}[(1-\xi_t)^m].
\]
Hence, if we now switch to the case $\gamma\le 0$, 
\[
\E^\gamma_{1-x}[\xi_t^m]=\E^{-\gamma}_m[(1-x)^{A_t}].
\]
This relation applied to the representation of the non-equilibrium AFS
in (\ref{eq:thmmainlimit2}) yields
\[
\Pb^\gamma_{1-y}(\tau_1\le  t)
=\Pb^\gamma_{1-y}(\xi_t=1)
=\lim_{m\to\infty} \E^\gamma_{1-y}(\xi_t^m)
=\E^{-\gamma}_\infty[(1-y)^{A_t}]
\]
so that 
\[
\frac{\Pb^\gamma_{1-y}(\tau_1\le t)}{1-y}
=  \E^{-\gamma}_\infty[(1-y)^{A_t-1}]
\]
and we obtain the following alternative representation of the
non-equilibrium AFS in Theorem 1. 

\begin{corollary}\label{cor:main}
For the case of neutral evolution or negative selection, $\gamma\le 0$, the
non-equilibrium AFS in Theorem \ref{thm:main} has the representation
\[ 
\omega_\gamma\int_0^1 f(y) \,\Pb^{*\gamma}_{1-y}(\tau_1\le t) \pi_\gamma(y)\,dy
=\int_0^1 \E^{-\gamma}_\infty[(1-y)^{A_t-1}]f(y) \, \pi_0(y)\,dy
\] 
\end{corollary}
As a remark on the computational aspects of this representation, the
probabilities $\Pb^0_\infty(A_t=j)$ for the case $\gamma=0$ are those
of (\ref{eq:distrpassage},\ref{eq:distrkingman}) in the limit
$m\to\infty$.  For the negative case $\gamma<0$, one approach might be
averaging over a large number of simulated paths of $(A_t)$.

\section{Population functionals and sample functionals}

At this stage, we have used the discrete time Markov chain
$(X_n)$, the time-scaled allele frequencies $(X_N(t))$, the continuous
time and continuous state scaled version $(Y_{N,L}(t))$, the random
field $(\Xc_t^{N,L})$, and the Poisson stochastic integral
$(\Xc_t^N)$, to study non-equilibrium allele frequencies. Moreover,
stationary versions $\Xc_\infty^N$ and $\Xc_\infty$ appear in the
large time limit. Figure \ref{fig:scheme} shows the relation of the
various random fields and the corresponding expectations. The same
sequence of approximations apply to building other functionals of the
allele frequencies, where it is natural to distinguish between
population functionals and sample functionals. Population functionals
are in principal non-observable and require knowing the history of the
entire spectrum of allele frequencies in each site counted as
fractions of the entire population. Sample functionals refer to the spectrum
of frequencies being restricted to a smaller sample of individuals, in the
sense of fixing an integer $m\ge 1$ and consider a sample of $m$
sequences chosen randomly with equal probabilities for all subsets of size $m$.

\begin{figure}\label{fig:scheme}
\begin{tikzpicture}
 \matrix (m) [matrix of math nodes,
    row sep = 2em, column sep=1em, minimum width=7em, minimum height=4em]
  {  
          \langle \Xc_t^{N,L}\!,f\rangle 
          & {} 
          & {} \\
          \langle \Xc_t^N\!,f\rangle 
          &  \langle \Xc_\infty^N,f\rangle 
          & \langle \Xc_\infty,f\rangle \\
          \E^\gamma\langle \Xc_t^N\!,f\rangle 
          &  \E^\gamma\langle \Xc_\infty^N,f\rangle
          &  \E^\gamma\langle \Xc_\infty,f\rangle  \\ 
          \theta\omega_\gamma\int_0^1\Pb^{*\gamma}_{1-y}(\tau_1\le
               t)f(y)\pi_\gamma(y)\,dy\quad 
          & {} 
          &  \quad\theta\omega_\gamma\int_0^1 f(y)\pi_\gamma(y)\,dy\\ 
      };
  \path[-stealth]
    (m-1-1) edge [double] node [right] {$L\to\infty$} (m-2-1)
    (m-2-1) edge node [left] {} (m-3-1)
            edge [double] node [below] {$t\to\infty$} (m-2-2)
    (m-2-2) edge [double] node [below] {$N\to\infty$} (m-2-3)
            edge   node   {} (m-3-2)
    (m-2-3) edge node [right] {} (m-3-3)           
    (m-3-1) edge [double]   node [right] {$N\to\infty$} (m-4-1)
            edge [double] node [below] {$t\to\infty$} (m-3-2)
    (m-3-2) edge [double] node [below] {$N\to\infty$} (m-3-3)
            edge [double] node [left] {$N\to\infty$} (m-4-3)
    (m-3-3)    edge node [below] {}   (m-4-3)
    (m-4-3)    edge node [below] {}   (m-3-3)
    (m-4-1)  edge [double] node [below] {$t\to\infty$} (m-4-3)
;
\end{tikzpicture}
\caption{Schematic structure of the random fields (upper two rows),
  their expectations (third row) and limiting expectations (third and
  fourth row). Going from left to right the graph indicates the
  various non-equilibrium and equilibrium versions.}
\end{figure}

To identify the proper sampling functionals we consider the discrete
generation version of the population model. 
Consider a sample of $m$ sequences. Let
\[
M^i_n=\mbox{\# of sampled derived alleles in site $i$ of
  generation $n$}. 
\]
Conditionally, for each $n$ given $X_n$, the random variables $M_n^i$,
$1\le i\le L$, are i.i.d.\ and suitably approximated by the binomial
distribution Bin$(m,X^i_n/N)$. With these insights drawn from the
discrete generation model it follows that $M_{[Nt]}^i$ is binomial
with parameters $m$ and $N^{-1}X^i_{[Nt]}$, where the latter is a
discrete approximation of $Y^i_{N,L}(t)$. Considering the map
\[
t\mapsto \langle\Mc^{N,L}_t\!,g\rangle
=\sum_{i=1}^L g(M_{N,L}^i(t)),\quad M_{N,L}^i(t)\sim 
\mbox{Bin}(m,Y_{N,L}^i(t)),
\]    
for $g:\Z\to \R$, $g(0)=0$, in analogy with Proposition
\ref{prop:Ltoinfty} we obtain as $L\to\infty$ a limiting random field
$\{\langle\Mc^N_t\!,g\rangle,\;t\ge 0\}$, such that
\[
\Mc^N_t=\sum_{k=1}^m \delta_{Z^{N,k}_t}
\]
and the family of random variables
\[
Z^{N,k}_t
=\mbox{\# of sites with exactly $k$ derived in the sample},
\quad 1\le k\le m, 
\]
are independent and Poisson distributed with mean 
\[
\E^\gamma(Z_t^{N,k})=
\theta \binom{m}{k}\, N\E^\gamma_{1/N}\Big[\int_0^{t\wedge\tau} \xi_u^k(1-\xi_u)^{m-k}\,du\Big].
\]
Moreover, the sample functional $\Mc_t^N$ has a stationary version
$\Mc_\infty^N$ such that the family $Z^{N,k}_\infty$, $1\le k\le m$,
are independent Poisson with mean
\[
\E^\gamma(Z_\infty^{N,k})=
\theta \binom{m}{k}\, N\E^\gamma_{1/N}\Big[\int_0^\tau \xi_u^k(1-\xi_u)^{m-k}\,du\Big].
\]
It is outside the scope of this work to study any limiting family of
random processes, $\{(Z^k_t)_{t\ge 0},1\le k\le m\}$, that might arise as
$N\to\infty$.  As a consequence of Theorem \ref{thm:main}, however, we
do obtain the following additional results.  As $N\to\infty$, the
Poisson expectations have limits
\begin{align}\label{eq:segsites_k}
\E^\gamma(Z_t^{N,k})\to 
\theta \binom{m}{k}\,\omega_\gamma \int_0^1
\Pb^{*\gamma}_{1-y}(\tau_1\le t)y^k(1-y)^{m-k}\pi_\gamma(y)\,dy
\end{align} 
and
\[
\E^\gamma(Z_\infty^{N,k})\to 
\theta \binom{m}{k}\, \omega_\gamma\int_0^1 y^k (1-y)^{m-k}\,\pi_\gamma(y)\,dy.
\]

As an application we fix the sample size $m$ and consider 
\[
Z_m^N(t)=\sum_{k=1}^{m-1} Z^{N,k}_t=\mbox{\# of segregating
  sites in sample of size $m$}. 
\] 
Using notation $a_m$ for the binomial formula 
\[
a_m(y)=\sum_{j=1}^{m-1} \binom{m}{j} y^j(1-y)^{m-j}= 1-y^m-(1-y)^m
,\quad 0\le y\le 1,
\]
it follows that the summation $Z_m^N(t)$ is Poisson distributed with
expected value
\[
\E^\gamma Z_m^N(t)=N\theta 
\E^\gamma_{1/N}\Big[\int_0^{t\wedge\tau} a_m(\xi_u)\,du\Big].
\]
Since $a_m\in \Fc_0$, the limiting expected number of segregating sites,
\[
S_m(t) =\lim_{N\to\infty}\E^\gamma Z_m^N(t),
\] 
is now obtained from Theorem \ref{thm:main} as
\[
S_m(t)=\theta \omega_\gamma
\int_0^1 \Pb^{*\gamma}_{1-y}(\tau_1\le t)a_m(y) \pi_\gamma(y)\,dy.
\]
Moreover, if $\gamma\le 0$ then by Corollary \ref{cor:main} 
\[
S_m(t)
=\theta \int_0^1 E^{-\gamma}_\infty[(1-y)^{A_t-1}]\,a_m(y) \, \pi_0(y)\,dy,
\]
and hence, by evaluating the integral,
\begin{align*}
S_m(t)&=2\theta\,E^{-\gamma}_\infty\Big[\sum_{k=0}^{m-1}\frac{1}{k+A_t}-\frac{\Gamma(A_t)\Gamma(m)}{\Gamma(A_t+m)}\Big]\\
&\approx 2\theta  E^{-\gamma}_\infty[\log(1+(m-1)/A_t)].
\end{align*}
Under neutral evolution $\gamma=0$, one may use
(\ref{eq:kimurafixation}) or (\ref{eq:distrkingman}) to obtain series
representations of the limiting expressions of $S_m(t)$. Yet another
representation of the same quantity follows from
\begin{align*}
\E^0 Z_m^N(t)&=\theta\int_0^tN\E^0_{1/N}[a_m(\xi_u)]\,du\\
  &=\theta\int_0^t N\E_m^0[1-(1/N)^{A_u}-(1-1/N)^{A_u}]\,du\\
&\to \theta\int_0^t (\E^0_m (A_u)-\Pb^0_m(A_u=1))\,du,\quad N\to\infty.
\end{align*}
Thus, by (\ref{eq:distrpassage}) and (\ref{eq:expkingman}),
\[
S_m(t)=\theta \sum_{k=2}^m (1+(-1)^k)(1-e^{-\binom{k}{2}t})
\frac{2k-1}{\binom{k}{2}}\frac{m_{[k]}}{m_{(k)}},
\]
and we recover the fixed population size version of an expression
which can be found in \cite{Tajima1989} and \cite{Zivkovic2011}. 
For $\gamma\not=0$, the seemingly crude but straightforward approximation  
\begin{equation}\label{eq:approxbyneutral}
S_m(t)\approx \theta \omega_\gamma
\int_0^1 \Pb^{*0}_{1-y}(\tau_1\le t)a_m(y) \pi_\gamma(y)\,dy,
\end{equation}
obtained by simply replacing the conditional distribution in
(\ref{eq:thmmainlimit}) with its neutral, explicitly known, version
(\ref{eq:kimurafixation}), appears quite efficient and useful for
many purposes.   In steady-state, letting $t\to\infty$,
\[
S_m(t)\to S_m= \theta \omega_\gamma\int_0^1 a_m(y)\,\pi_\gamma(dy).
\]
The neutral case $\gamma=0$ yields the familiar relation
\[
S_m= 2\theta \sum_{k=1}^{m-1}\frac{1}{k}\approx 2\theta \ln m.
\]
As an illustration of these findings, Figure \ref{fig:segsites} shows
the non-equilibrium growth over time of the limiting expected number
of segregating sites.

\begin{figure}
\includegraphics[width=0.5\textwidth]{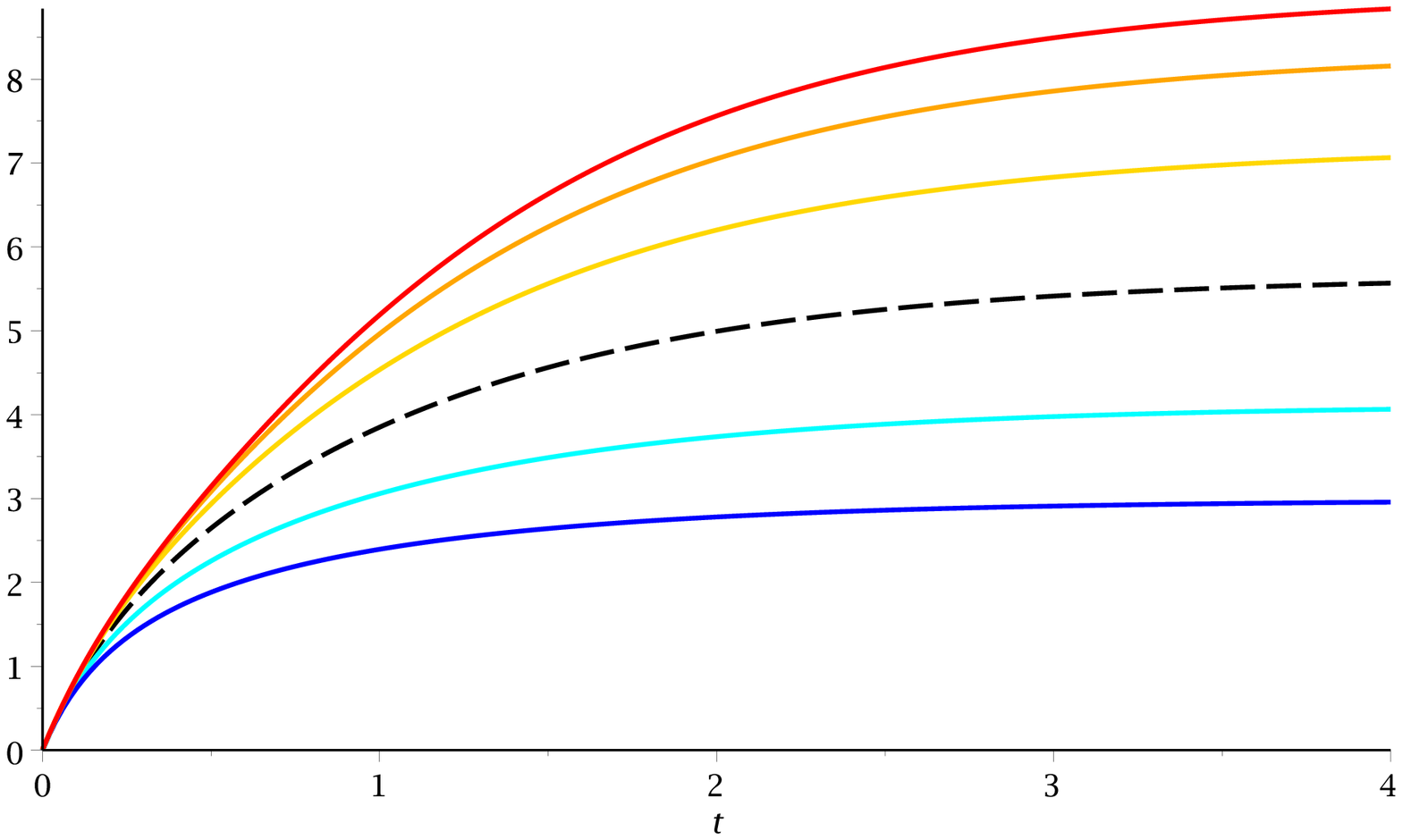}
\includegraphics[width=0.5\textwidth]{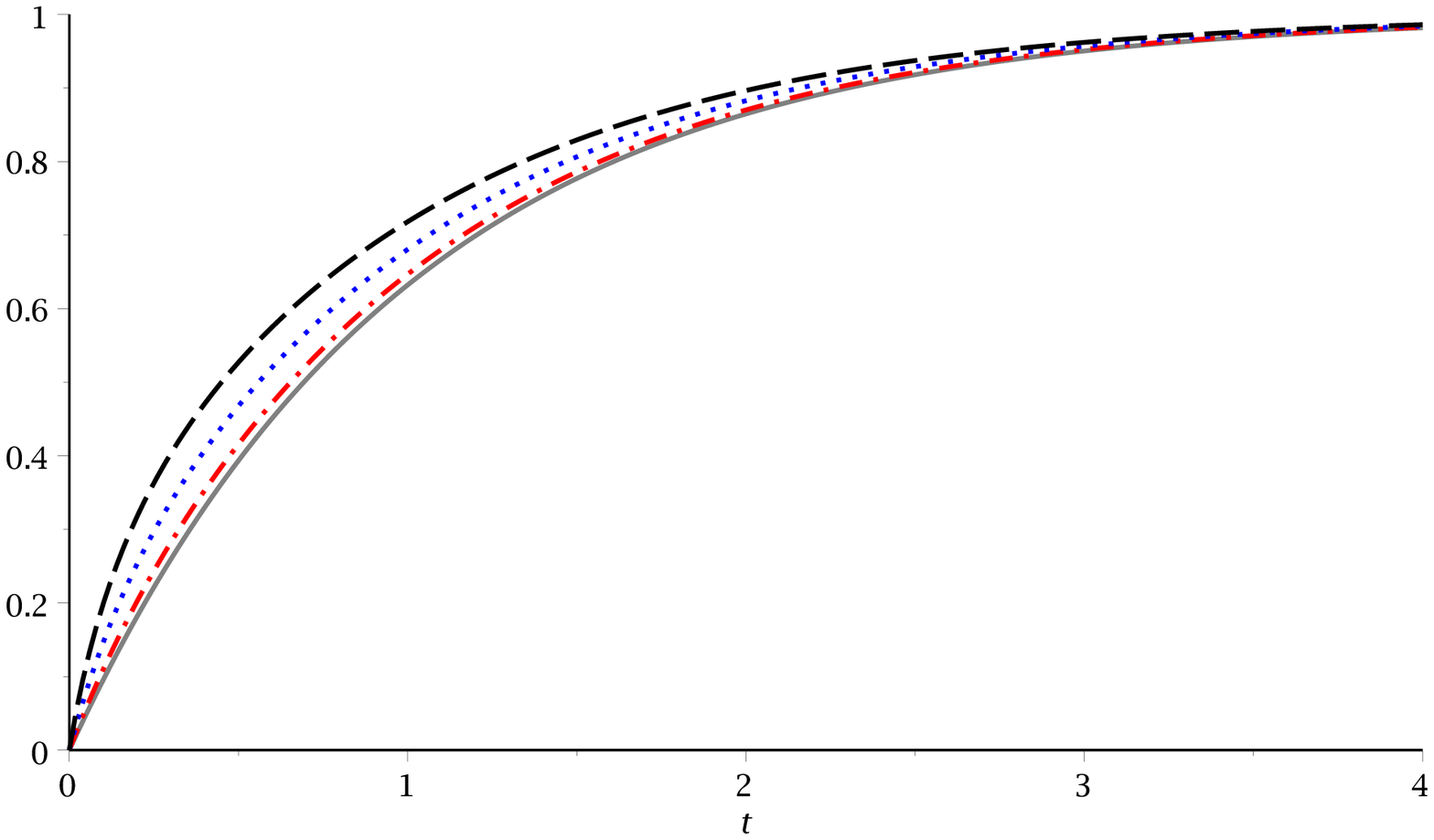}
\caption{Left panel: the number of segregating sites, $\{S_m(t),
  0\le t\le 4\}$, in a sample of $m=10$ individuals, using the
  approximation (\ref{eq:approxbyneutral}), for six scenarios of
  selection, $\gamma=-2$ (blue solid line), $\gamma=-1$ (turquoise
  solid line), $\gamma=0$ (black dashed line), $\gamma=1$ (yellow
  solid line), $\gamma=2$ (orange solid line) and $\gamma=3$ (red
  solid line).  Right panel: illustration of the effect of sample size
  for the neutral case $\gamma=0$. The normalized number of
  segregating sites, $\{S_m(t)/S_m, 0\le t\le 4\}$,  $m=2$
  (green solid line), $m=5$ (red dashed-dotted line),
  $m=10$ (blue dotted line), and $m=20$
  (black dashed line).  }
\label{fig:segsites}
\end{figure}

\section{Discussion} 
Several allele frequency-based summary statistics are central
measurements in population genetic studies. In the study of a single
population allele frequency-based measurements assist, for example,
the identification of candidate loci of adaptive evolution
\citep{Tajima1989,Fay2000,Zeng2006}, or the inference of the
demographic history of the population \citep{Excoffier2013}.  In the
study of speciation events, the consideration of the joint AFS of
closely related species is of considerable interest for the underlying
mechanics of speciation \citep{Gutenkunst2009}.  Moreover, measures of
population differentiation, which depend on the comparison of the AFS
between two populations descending from a common ancestor, are
commonly used to detect genomic regions involved in the process of
speciation \citep{Seehausen2014}. From a theoretical viewpoint these
kind of inferences require a sound analytical understanding of the
non-equilibrium properties of the AFS as a function of time starting
from any specific initial value.

\cite{Evans2007} consider the non-equilibrium AFS in a population of
size $2N\rho$ with mutation rate $\theta'$ in a bi-allelic setting
with allele frequencies governed by a general class of diffusion
processes, and where asymptotically in the scaling limit $N\to\infty$
the factor $\rho=\rho(t)$ may depend on time.  In this setting they
derive a limiting function $f(t,x)$ such that
\[
f(t,x)\,dx \sim \mbox{expected fraction of alleles
  of frequency $[x,x+dx]$},\quad 0<x<1.
\]
For the case of the Wright-Fisher diffusion process with selection
parameter $S$ and varying population size $\rho(t)$, the function
$f(t,x)$ is characterized by the property that the scaled function
$g(t,x)=x(1-x) f(t,x)$ is a solution of
\[
\frac{\partial }{\partial t} g(t,x)=
-S x(1-x) \frac{\partial }{\partial x} g(t,x)
+\frac{x(1-x)}{2\rho(t)}\frac{\partial^2 }{\partial x^2} g(t,x)
\]
with appropriate initial value $g(0,x)=g_0(x)$ at $t=0$ and boundary conditions
\[
\lim_{x\to 0} g(t,x)=\theta' \rho(t),\quad 
\lim_{x\to 1} g(t,x)=0.
\]
Comparing this result for the case $\rho(t)=1$, $S=\gamma$,
$\theta'=2\theta$ and $g_0=0$ with our representation
(\ref{eq:thmmainlimit2}) of the AFS in Theorem 1, leads to identifying
the implicitly given function $g$ as
\[
g(t,x)=2\theta \,\Pb^\gamma_{1-x}(\tau_1\le t).
\]
Similarly, using (\ref{eq:thmmainlimit}), one obtains the representation
\[
f(t,x)=2\theta \,\Pb^{*\gamma}_{1-x}(\tau_1\le
t)\,\frac{1-e^{-2\gamma (1-x)}}{x(1-x)(1-e^{-2\gamma})},
\]
reminding that $\Pb_{1-x}^\gamma$ is the probability law of the
Wright-Fisher diffusion process with initial state $1-x$, $\tau_1$ is
the time to fixation in state $1$ of the same diffusion, and
$\Pb^{*\gamma}_{1-x}$ is the conditional law given $\tau_1<\infty$.
In particular, for the neutral case $\gamma=0$,
\[
f(t,x)=2\theta\, \Pb^{*0}_{1-x}(\tau_1\le t)/x
\] 
with the probability distribution function given explicitly in
(\ref{eq:kimurafixation}). For non-positive selection, $\gamma\le 0$,
using the framework of duality theory of moment functionals, and
recalling the birth-death process $(A_t)$ with intensities given in
(\ref{eq:kingmanintensity}) and (\ref{eq:branchingintensity}), an
alternative representation is given by
\[
f(t,x)= \frac{2\theta \,\E^{-\gamma}[(1-x)^{A_t}|A_0\to \infty]}{x(1-x)}
\]
in terms of $A_t$ ``coming down from infinity'' at time $t=0$. 
We thus provide exact analytical results on the non-equilibrium AFS,
which if efficiently applied can improve our understanding of
a broad spectrum of population genetic inferences.

Our framework can further be related to the computational method devised by
\cite{Evans2007}, which uses the fact that the moments
\[
\mu_n(t)=\int_0^1 x^n g(t,x)\,dx,\quad n=0,1,\dots
\]
satisfy a coupled system of ordinary differential equations. Assuming
constant population size, these moments can be identified within our
framework by 
\[
\mu_n(t)=\lim_{N\to\infty} \E\langle \Xc_t^N\!,b_n\rangle,\quad
b_n(y)=(1-y)y^{n+1}.
\]
Since $b_n\in \Fc_0$, by Theorem \ref{thm:main}, 
\[
\mu_n(t)=\theta\omega_\gamma\int_0^1 b_n(y) \,\Pb^{*\gamma}_{1-y}(\tau_1\le t)
\pi_\gamma(y)\,dy.
\]

Moreover, the analytical description of the build-up of the AFS
starting from a completely mono-allelic state at $t=0$, can be helpful
to study the build-up of lineage-specific polymorphisms as compared to
the ebbing of shared ancestral polymorphisms during the process of
speciation.  Our work therefore ties to the work of \cite{Wakeley1997}
and \cite{Chen2012}, who model the separate AFS of lineage-specific
and shared ancestral polymorphisms in samples from closely related
species in the framework of coalescent theory.  Precisely, our
relation (\ref{eq:segsites_k}) corresponds to the AFS of
lineage-specific polymorphisms in \cite{Chen2012} (5), which in our
notation reads
\[
\lim_{N\to\infty} \E(Z_t^{N,k})= \theta\,\E^0_m\Big[\sum_{i=A_t}^m
\frac{\binom{m-k-1}{i-2}}{\binom{m-1}{i-1}}i\,\E(T_i|A_t)\Big],
\]
where $T_i$ is the time interval during which $i$ lineages exist.
While \cite{Chen2012} provides results on the neutral AFS and the AFS
in regions that underwent selective sweeps, we here add the case of
negative selection.  More importantly, our work illustrates that the
use of the duality relation between the Wright-Fisher diffusion
process and a class of birth-death processes can tie several
frameworks together.  We therefore foresee a broad applicability of
the framework presented in this study.

\section{Technical details and remaining proofs}
\label{sec:technical}

In order to find the limits of $\Xc^{N,L}$ and $\Mc^{N,L}$ as
$L\to\infty$, is is convenient to use a method first devised for
Poisson random balls models \citep{Kaj2007}, which applies a set of
signed measures for indexing.  

\subsection{Indexing random fields by measures}

Let $\Hc$ be the set of finite, signed measures on the positive
real line, let $|\mu|$ denote the variation norm on $\Hc$, and put
\[
\Hc_0=\{\mu\in \Hc: \int_0^\infty |\mu|(du)<\infty\},\quad 
\Hc_1= \{\mu\in \Hc: \int_0^\infty  u \,|\mu|(du)<\infty\}.
\]
For a given $f\in\Fc$, denote
\[
\langle \Xc^{N,L},\mu\rangle
=\int \langle \Xc_u^{N,L}\!,f\rangle \,\mu(du)
=\sum_{i=1}^L \int f(Y_{N,L}^{(i)}(u))\,\mu(du),\quad \mu\in\Hc_0.
\]
In particular, with $\mu=\delta_t$ we have $\langle
\Xc^{N,L},\delta_t\rangle=\langle\Xc^{N,L}_t\!,f\rangle$. The quantity
$\langle \Mc^{N,L},\mu\rangle$ is defined analogously.  Similarly, for
$\mu\in \Hc$,
\[
\langle \Xc^N\!,\mu\rangle=\int \langle \Xc_u^N\!,f\rangle\,\mu(du)
=\int_{\R^+\times\Dc} 
\int f(\xi^s_u)\,\mu(du)\,\Nc_N(ds,d\xi^s).
\]
Now, in greater generality than Proposition \ref{prop:stochint},
$\langle \Xc^N\!,\mu\rangle$ is well-defined with finite expected
value, such that, for every $f\in \Fc$,
\begin{equation}\label{eq:indexmu1}
\E \langle\Xc^N\!,\mu\rangle= 
\int N\theta \,\E^\gamma_{1/N}\Big[\int_0^u f(\xi^0_v)
  \,dv\Big]\,\mu(du)<\infty, \quad \mu\in \Hc_1,
\end{equation}
and for $f\in \Fc_0$,  
\begin{equation}\label{eq:indexmu2}
\E \langle\Xc^N\!,\mu\rangle= 
\int N\theta\, \E^\gamma_{1/N}\Big[\int_0^{u\wedge \tau} f(\xi^0_v)
  \,dv\Big]\,\mu(du)<\infty, \quad \mu\in \Hc_0.
\end{equation}
Indeed, to verify (\ref{eq:indexmu1}) it suffices to show
\[
\int_{\R^+\times\Dc}  \Big|\int
f(\xi^s_u)\,\mu(du)\Big|\,n_N(ds,d\xi^s)<\infty.
\]
Here
\begin{align*}
\int_{\R^+\times\Dc} \int
|f(\xi^s_u)|\,|\mu|(du)\,n_N(ds,d\xi^s)
&=\E^\gamma_{1/N}\int_0^\infty \int|f(\xi^s_u)|\,|\mu|(du)\,N\theta ds\\
&= N\theta \int
\E^\gamma_{1/N}\int_0^\infty|f(\xi^s_u)|\,ds\,|\mu|(du),
\end{align*}
and for fixed $u$, 
\[
\E^\gamma_{1/N}\int_0^\infty|f(\xi^s_u)|\,ds
=\E^\gamma_{1/N}\int_0^u|f(\xi^s_u)|\,ds
=\E^\gamma_{1/N}\int_0^u|f(\xi^0_v)|\,dv.
\]
For $f\in\Fc$ this implies
\[
\int \E^\gamma_{1/N}\int_0^\infty|f(\xi^s_u)|\,ds\,|\mu|(du) 
\le  \|f\|_\infty \int u\,|\mu|(du) <\infty,\quad \mu\in \Hc_1,
\]
and so 
\begin{align*}
\E \langle\Xc^f_N,\mu\rangle &= 
\int_{\R^+\times\Dc}  \int f(\xi^s_u)\,\mu(du)\,n_N(ds,d\xi^s)\\
&=\int  N\theta\,\E^\gamma_{1/N}\Big[\int_0^u f(\xi^0_v)
  \,dv\Big]\,\mu(du),\quad \mu\in \Hc_1.
\end{align*}
To show (\ref{eq:indexmu2}) we assume $f\in \Fc_0$. Then, for
fixed $u$,
\[
\E^\gamma_{1/N}\int_0^\infty|f(\xi^s_u)|\,ds
=\E^\gamma_{1/N}\int_0^{u\wedge \tau}|f(\xi^0_v)|\,dv
\le \E^\gamma_{1/N}\int_0^\tau|f(\xi^0_v)|\,dv 
\]
and hence using (\ref{eq:scaledoccupationtime}),
\[
N\theta \int
\E^\gamma_{1/N}\int_0^\infty|f(\xi^s_u)|\,ds\,|\mu|(du)
\le \theta \sup_N N\E^\gamma_{1/N}\int_0^\tau|f(\xi^0_v)|\,dv
\,\int|\mu|(du)<\infty.
\]

\subsection{Proof of Proposition \ref{prop:Ltoinfty}}

To demonstrate the convergence in finite dimensional
distributions of the sequence $\{\langle\Xc^{N,L}_t\!,f\rangle,t\ge 0\}$ 
to $\{\langle\Xc^N_t\!,f\rangle,t\ge 0\}$ 
we need to show the convergence in distribution
\[
\sum_{k=1}^n\alpha_k \langle\Xc^{N,L}_{t_k}\!,f\rangle \stackrel{d}{\implies}
  \sum_{k=1}^n\alpha_k \langle\Xc^N_{t_k} \!,f\rangle 
\]
for arbitrary weights $\alpha_1,\dots,\alpha_n$ and arbitrary time points
$t_1\le \dots\le t_n$, $n\ge 1$.   But obviously, 
letting $\mu_n=\sum_{k=1}^n \alpha_k\delta_{t_k}$, this is the same as
the convergence in distribution of $\langle\Xc^{N,L}\!,\mu_n\rangle$
to $\langle\Xc^N\!,\mu_n\rangle$.  Because of (\ref{eq:indexmu1}) and
(\ref{eq:indexmu2}) we have control of the generalized Poisson functionals 
in the limit and, therefore, we may continue with the method of moment
generating functions.  Specifically, using the defining properties of
Poisson random measures, 
\[
\ln \E \exp\{\alpha \langle\Xc^N\!,\mu\rangle \}= N\theta
\,\E_{1/N}^\gamma\Big[\int_0^\infty (e^{\alpha \int
    f(\xi^s_u)\,\mu(du)}-1)\,ds\Big],
\]
where $\mu\in \Hc_1$ or $\Hc_0$ in line with using either
(\ref{eq:indexmu1}) or (\ref{eq:indexmu2}).

On the other hand, for $\Xc^{N,L}$, by the independence of the sites, 
\begin{align}\nonumber
\ln \E \exp\{\alpha \langle\Xc^{N,L}\!,\mu\rangle \}
& =  L\ln \E \exp\{\alpha \int f(Y^{(1)}_{N,L}(u))\,\mu(du) \}\\
& \sim  L\,\E\Big[ \exp\Big\{\alpha \int
  f(Y^{(1)}_{N,L}(u))\,\mu(du)\Big\}-1\Big],  
\label{eq:momentgenfunct}
\end{align}
where we use the notation $A_L\sim B_L$ for $A_L/B_L\to 1$ as
$L\to\infty$.  Let $\{\kappa_j\}_{j\ge 1}$ be independent, exponential random
variables with intensity $N\theta/L$. Since $f(0)=0$,
\[
 \int f(Y^{(1)}_{N,L}(u))\,\mu(du)
=\int_{\kappa_1}^{\kappa_1+\tau}f(\xi_u^{\kappa_1})\,\mu(du) 
+\int_{\kappa_1+\tau+\kappa_2}^\infty f(Y^{(1)}_{N,L}(u))\,\mu(du).
\]
The expected value over $\kappa_1$ on the right hand side in
(\ref{eq:momentgenfunct}) now evaluates to 
\begin{align*}
\int_0^\infty \theta N e^{-\theta Ns/L}
\E\Big[ \exp\Big\{\alpha \int_s^{s+\tau}
  f(\xi_u^s)\,\mu(du)+\alpha\int_{s+\tau+\kappa_2}^\infty \dots\,\mu(du) \Big\}-1\Big]\,ds.
\end{align*}
Letting $L\to\infty$ we have $e^{-\theta Ns/L}\to 1$ and $\kappa_2\to\infty$, and so
\[
\ln \E \exp\{\alpha \langle\Xc^{N,L}\!,\mu\rangle \}
 \to   \,\int_0^\infty \theta N 
\E_{1/N}^\gamma\Big[ \exp\Big\{\alpha \int_s^{s+\tau}
  f(\xi_u^s)\,\mu(du)\Big\}-1\Big]\,ds,
\]
which is the logarithmic moment generating function of
$\langle\Xc^N\!,\mu\rangle$, hence completing the proof of convergence
in distribution. 


\subsection{Extending Theorem \ref{thm:main} from $\Fc_0$ to $\Fc$}

Our proof of Theorem \ref{thm:main} is stated for $f\in \Fc_0$, hence
specifically functions $f$ with $f(1)=0$. For some applications,
however, it is more natural to work with the class $\Fc$ allowing
$f(1)\not= 0$.  To account for such cases we included two extended
versions of (\ref{eq:thmmainlimit}) in the theorem.  These follow
immediately from (\ref{eq:thmmainlimit}) together with relations
(\ref{lem:extend2}) and (\ref{lem:extend3}) of the following lemma. 

\begin{lemma} \label{lem:extendmain}
Let $f$ be a bounded real-valued function defined on
  $[0,1]$ such that $f(0)=0$.  Then  
\begin{align}\label{lem:extend1}
&\E_x^\gamma \int_0^t f(\xi_u)\,du 
=\E_x^\gamma \int_0^{t\wedge\tau} f(\xi_u)\,du
+f(1)q_\gamma(x)\int_0^t \Pb_x^{*\gamma} (\tau_1\le u)\,du,\\[2mm]
\nonumber
&\E_x^\gamma \int_0^{t\wedge\tau} f(\xi_u)\,du 
  =\E_x^\gamma \int_0^{t\wedge\tau} 
(f(\xi_u)
-f(1)q_\gamma(\xi_u))\,du\\
&\qquad +f(1)q_\gamma(x)\int_0^t \Pb_x^{*\gamma} (\tau_1>u)\,du,
\label{lem:extend2}\\[2mm]
\label{lem:extend3}
&\E_x^\gamma \int_0^t f(\xi_u)\,du 
=\E_x^\gamma \int_0^{t\wedge\tau} (f(\xi_u)-f(1)q_\gamma(\xi_u))\,du
+f(1)q_\gamma(x)t.
\end{align}

\end{lemma}

\begin{proof}
For bounded functions $f$ on $[0,1]$,
\[
\E^\gamma_x\int_0^{t\wedge\tau}f(\xi_r)\,dr
=\int_0^t \E^\gamma_x[f(\xi_r),\tau>r]\,dr
\]
and
\begin{align*}
\E^\gamma_x
f(\xi_s)&=\E^\gamma_x[f(\xi_s),\tau_0<s]+\E^\gamma_x[f(\xi_s),\tau_1<s]
+\E^\gamma_x[f(\xi_s),\tau>s]\\
&=f(0) \Pb^\gamma_x(\tau_0<s)+f(1)
\Pb^\gamma_x(\tau_1<s)+\E^\gamma_x[f(\xi_s),\tau>s],
\end{align*}
hence
\begin{align*}
\E_x^\gamma \int_0^t& f(\xi_u)\,du\\ 
&=\E_x^\gamma \int_0^{t\wedge\tau} f(\xi_u)\,du 
+f(0)\int_0^t \Pb_x^\gamma (\tau_0\le u)\,du 
+f(1)\int_0^t \Pb_x^\gamma (\tau_1\le u)\,du\\
&=\E_x^\gamma \int_0^{t\wedge\tau} f(\xi_u)\,du
+f(0)(1-q_\gamma(x))
+f(1)q_\gamma(x)t\\
&\quad -f(0)(1-q_\gamma(x))\int_0^t \Pb_x^{*\gamma}(\tau_0> u)\,du
-f(1)q_\gamma(x)\int_0^t \Pb_x^{*\gamma} (\tau_1>u)\,du.
\end{align*}
Take $f(0)=0$ to obtain statement (\ref{lem:extend1}). Also
\begin{align*}
\E_x^\gamma \int_0^{t\wedge\tau}\!\! f(\xi_u)\,du 
&=\E_x^\gamma \int_0^t f(\xi_u)\,du-f(1)q_\gamma(x)t
 +f(1)q_\gamma(x)\int_0^t \Pb_x^{*\gamma} (\tau_1>u)\,du.
\end{align*}
Using $\E^\gamma_x[q_\gamma(\xi_u)]=q_\gamma(s)$, this may be
rewritten
\begin{align*}
\E_x^\gamma \int_0^{t\wedge\tau} f(\xi_u)\,du 
&=\E_x^\gamma \int_0^t (f(\xi_u)-f(1)q_\gamma(\xi_u))\,du\\
&\quad  +f(1)q_\gamma(x)\int_0^t \Pb_x^{*\gamma} (\tau_1>u)\,du.
\end{align*} 
Therefore, since the function $g$ defined by
$g(y)=f(y)-f(1)q_\gamma(y)$ satisfies $g(0)=g(1)=0$, an application of
(\ref{lem:extend1}) to $g$ implies the second statement (\ref{lem:extend2}). Combine (\ref{lem:extend1}) and (\ref{lem:extend2}) to
get (\ref{lem:extend3}).
\end{proof}

\section{Acknowledgments}
The authors are grateful to Hans Ellegren for encouraging this work.

\newpage
\section*{References}
\bibliography{cites_selection}{}

\end{document}